\documentclass[aps,pra,twocolumn,superscriptaddress,notitlepage,nopacs,amsmath,amstex,amssymb,citeautoscript,longbibliography,floatfix,twocolumn]{revtex4-2}

\usepackage{enumitem}
\usepackage{natbib}
\usepackage[english]{babel}
\usepackage{letltxmacro}
\usepackage{latexsym}
\LetLtxMacro{\ORIGselectlanguage}{\selectlanguage}
\makeatletter
\DeclareRobustCommand{\selectlanguage}[1]{%
  \@ifundefined{alias@\string#1}
    {\ORIGselectlanguage{#1}}
    {\begingroup\edef\x{\endgroup
       \noexpand\ORIGselectlanguage{\@nameuse{alias@#1}}}\x}%
}
\newcommand{\definelanguagealias}[2]{%
  \@namedef{alias@#1}{#2}%
}
\makeatother
\definelanguagealias{en}{english}
\definelanguagealias{English}{english}

\newcommand{\be}{\begin{equation}}
\newcommand{\ee}{\end{equation}}
\newcommand{\bea}{\begin{eqnarray}}
\newcommand{\eea}{\end{eqnarray}}

\renewcommand{\vec}[1]{{\bf #1}}

\newcommand{\im}{{\rm i}}

\usepackage{amsthm}

\newtheorem{theorem}{Theorem}
\newtheorem{corollary}{Corollary}[theorem]

\usepackage{graphicx}
\usepackage{bm}
\usepackage{physics}
\usepackage[dvipsnames]{xcolor}
\usepackage{algorithm}
\usepackage{algpseudocode}
\usepackage{blindtext}
\usepackage{orcidlink}
\usepackage[normalem]{ulem}
\makeatletter
\newcommand{\printfnsymbol}[1]{%
  \textsuperscript{\@fnsymbol{#1}}%
}
\makeatother
\usepackage{hyperref}
\usepackage{soul}
\hypersetup{
    unicode=false,          % non-Latin characters in AcrobatÕs bookmarks
    pdftoolbar=false,        % show AcrobatÕs toolbar?
    pdfmenubar=true,        % show AcrobatÕs menu?
    pdffitwindow=false,     % window fit to page when opened
    pdfstartview={FitH},    % fits the width of the page to the window
    pdftitle={},    % title
    pdfauthor={Stefan H. Sack, Raimel A. Medina, Richard Kueng, Maksym Serbyn},     % author
    pdfsubject={Recursive greedy initialization of the quantum approximate optimization algorithm with guaranteed performance improvement},   % subject of the document
    pdfcreator={},   % creator of the document
    pdfproducer={}, % producer of the document
    pdfkeywords={QAOA, transition states}, % list of keywords
    pdfnewwindow=true,      % links in new window
    colorlinks=true,       % false: boxed links; true: colored links
    linkcolor=teal,          % color of internal links (change box color with linkbordercolor)
    citecolor=teal,        % color of links to bibliography
    filecolor=teal,      % color of file links
    urlcolor=teal           % color of external links
}

\begin{document}
\title{Recursive greedy initialization of the quantum approximate optimization algorithm with guaranteed improvement}

\author{Stefan H. Sack  \orcidlink{0000-0001-5400-8508}}

\author{Raimel A. Medina \orcidlink{0000-0002-5383-2869}}
\affiliation{Institute of Science and Technology Austria (ISTA), Am Campus 1, 3400 Klosterneuburg, Austria}

\author{Richard Kueng \orcidlink{0000-0002-8291-648X}}
\affiliation{Institute for Integrated Circuits, Johannes Kepler University Linz, Altenberger Straße 69, Austria}

\author{Maksym Serbyn \orcidlink{0000-0002-2399-5827}}
\affiliation{Institute of Science and Technology Austria (ISTA), Am Campus 1, 3400 Klosterneuburg, Austria}

\date{\today}
\begin{abstract} 
The quantum approximate optimization algorithm (QAOA) is a variational quantum algorithm, where a quantum computer implements a variational ansatz consisting of $p$ layers of alternating unitary operators and a classical computer is used to optimize the variational parameters. For a random initialization, the optimization typically leads to local minima with poor performance, motivating the search for initialization strategies of QAOA variational parameters. Although numerous heuristic initializations exist, an analytical understanding and performance guarantees for large $p$ remain evasive. We introduce a greedy initialization of QAOA which guarantees improving performance with an increasing number of layers. Our main result is an analytic construction of $2p+1$ {\it transition states} --- saddle points with a unique negative curvature direction --- for QAOA with $p+1$ layers that use the local minimum of QAOA with $p$ layers.  Transition states connect to new local minima, which are guaranteed to lower the energy compared to the minimum found for $p$ layers. We use the \textsc{Greedy} procedure to navigate the exponentially increasing with $p$ number of local minima resulting from the recursive application of our analytic construction. The performance of the \textsc{Greedy} procedure matches available initialization strategies while providing a guarantee for the minimal energy to decrease with an increasing number of layers $p$. 
\end{abstract}
\maketitle

\section{Introduction}
The quantum approximate optimization algorithm (QAOA)~\cite{farhi2014quantum} is a prospective near-term quantum algorithm for solving hard combinatorial optimization problems on Noisy Intermediate-Scale Quantum (NISQ)~\cite{preskill2018quantum} devices. In this algorithm, the quantum computer is used to prepare a variational wave function that is updated in an iterative feedback loop with a classical computer to minimize a cost function (the energy expectation value), which encodes the computational problem. A common bottleneck of the QAOA is the convergence of the optimization procedure to one of the many low-quality local minima, whose number increases exponentially with the QAOA circuit depth $p$~\cite{zhou2018quantum, sack2021quantum}. 

Much effort has been devoted to finding good initialization strategies to prevent convergence to such low-quality local minima. Researchers have proposed to: first solve a relaxed classical optimization problem and to use that as an initial guess~\cite{egger2021warmstartingquantum}, to use machine learning  to infer patterns in the optimal parameters~\cite{jain2021graph}, 
interpolating optimal parameters between different circuit depths~\cite{zhou2018quantum}, or to use the parallels between the QAOA and quantum annealing~\cite{sack2021quantum}. 
Recently the success of the interpolation strategies that appeal to annealing was attributed to the ability of the QAOA to effectively speed up adiabatic evolution via the so-called counterdiabatic mechanism~\cite{wurtz2022counterdiabaticity}. This result was used to explain cost function concentration for typical instances and concentration of optimal, typically smoothly varying, parameters, which was previously introduced on Ref.~\cite{brandao2018fixed} and~\cite{akshay2021paramter} respectively. 

Despite this progress, all proposed initialization strategies remain heuristic or physically motivated at best, and our understanding of the QAOA optimization remains limited. One of the main puzzles is the exponential improvement of the QAOA performance with circuit depth $p$, observed numerically~\cite{zhou2018quantum,crooks2018performance}. Here we propose an analytic approach that relates QAOA properties at circuit depths $p$ and $p+1$. The recursive application of our result leads to a QAOA initialization scheme that guarantees improvement of performance with $p$. 

Our analytic approach relies on the consideration of stationary points of QAOA cost function beyond local minima. Inspired by the theory of energy landscapes~\cite{wales_2004}, we focus on stationary configurations with a unique unstable direction, known as \emph{transition states} (TS).
We show that $2p+1$ distinct TS can be constructed \emph{analytically} for a QAOA at circuit depth $p+1$ (denoted as QAOA$_{p+1}$) from minima at circuit depth $p$. All these TS for QAOA$_{p+1}$ exhibit the same energy as the QAOA$_{p}$-minimum from which they are constructed, 
thus providing a good initialization for QAOA$_{p+1}$.  Descending in the negative curvature direction connects each of the $2p+1$ TS to two local minima of QAOA$_{p+1}$, which are thus guaranteed to exhibit lower energy than the initial minima of QAOA$_{p}$. Iterating this procedure leads to an exponentially  increasing (in $p$) number of local minima which are guaranteed to have a lower energy at circuit depth $p+1$ than at $p$~\cite{Note1}. We visualize this hierarchy of minima and their connections in a graph and propose a \textsc{Greedy} approach to explore its structure. We numerically show that optimal parameters at every circuit depth $p$ are smooth (i.e. the variational parameters change only slowly between circuit layers) and directly connect to a smooth parameter solution at $p+1$ through the TS. Our results explain existing QAOA initializations and establish a recursive analytic approach to study QAOA.

The rest of the paper is organized as follows. In Section~\ref{sec:qaoa_landscape} we review the QAOA, present newly found symmetries, and introduce the analytical construction of TS. In Section~\ref{sec:ts_initialization} we show how TS can be used as an initialization to systematically explore the QAOA optimization landscape. From this, we introduce a new heuristic method, dubbed \textsc{Greedy} for exploring the landscape and provide a comparison to popular optimization strategies. Finally, in Section~\ref{sec:discussion} we discuss our results and potential future extensions of our work. Appendices~\ref{appx:symmetries}-\ref{App:last} present detailed proofs of our analytical results, as well as supporting numerical simulations. 

\section{QAOA optimization landscape} \label{sec:qaoa_landscape}
\subsection{MaxCut problem on random regular graphs}
The QAOA was originally proposed for a graph partitioning problem, known as finding the maximal cut (\textsc{MaxCut})~\cite{farhi2014quantum} and has also been applied to a variety of other optimization problems~\cite{streif2021beating, farhi2020needs, marwaha2021bounds}. 
\textsc{MaxCut} seeks for a partition of the given undirected graph $\mathcal{G}$ into two groups such that the number of edges $E$ that connect
vertices from different groups are maximized. 
Finding the solution of \textsc{MaxCut} for a graph with $n$ vertices is equivalent to finding a ground state for the $n$-qubit classical Hamiltonian $H_C = \sum_{\langle i,j\rangle \epsilon E} \sigma^z_i \sigma^z_j$, with the sum running over a set of graph edges $E$ and $\sigma^z_i$ being the Pauli-Z matrix acting on the $i$-th qubit. 

The depth-$p$ QAOA algorithm~\cite{farhi2014quantum} minimizes the expectation value of the classical Hamiltonian over the variational state $\ket{ \bm{\beta}, \bm{\gamma}}$ with angles $\bm{\beta}=(\beta_1,\ldots,\beta_p)$ and $\bm{\gamma}=(\gamma_1,\ldots,\gamma_p)$ shown in  Fig.~\ref{fig:1}(a):
\begin{equation}\label{eq:QAOA_ansatz}
    \ket{\bm{\beta}, \bm{\gamma}} = \prod_{i=1}^{p} e^{-\im \beta_i H_B} e^{-\im \gamma_i H_C} \ket{+}^{\otimes n}.
\end{equation}
Here $H_B = -\sum_i^n \sigma^x_i$ is the mixing Hamiltonian and the  circuit depth $p$ controls the number of applications of the classical and mixing Hamiltonian. 
The initial product state $\ket{+}^{\otimes n}$, where all qubits point in the $x$-direction is an equal superposition of all possible graph partitions which is also the ground state of $H_B$. 

Finding the minimum of $E(\bm{\beta}, \bm{\gamma})=\bra{\bm{\beta}, \bm{\gamma}} H_C \ket{\bm{\beta}, \bm{\gamma}}$ over angles 
$(\beta_1,\ldots,\beta_p)$ and 
$(\gamma_1,\ldots,\gamma_p)$
that form a set of $2p$ variational parameters,  $(\bm{\beta}, \bm{\gamma})$, yields a desired approximation to the ground state of $H_C$, equivalent to an approximate a solution of \textsc{MaxCut}. The scalar function $E(\bm{\beta}, \bm{\gamma})$ thus defines a $2p$-dimensional energy landscape where the QAOA  seeks to find the best minimum. The performance of the QAOA is typically reported in terms of the approximation ratio $r_{\bm{\beta}, \bm{\gamma}}=E(\bm{\beta}, \bm{\gamma})/C_{\min}$, where $
C_{\min}$ is the cost function value for the \textsc{MaxCut}. Symmetries of the QAOA ansatz when restricted to graphs with only odd connectivity, such as random 3-regular graphs (RRG3) used in this work, restrict the parameter range to the following fundamental region:
\begin{equation}\label{eq:symm}
 \beta_{i} \in \bigg[-\frac{\pi}{4},\frac{\pi}{4}\bigg]; 
 \ \,
 \gamma_1 \in \bigg(0,\frac{\pi}{4} \bigg),
 \ \, 
\gamma_{j} \in \bigg[-\frac{\pi}{4},\frac{\pi}{4} \bigg],
\end{equation}
with $i \in [1,p]$ and $j \in [2,p]$. Note that
the fundamental region presented above is smaller than what has been previously reported~\cite{zhou2018quantum, wang2018fermionic}, see the Appendix~\ref{appx:symmetries} for details.

\begin{figure}[tb]
    \centering
    \includegraphics[width=\columnwidth]{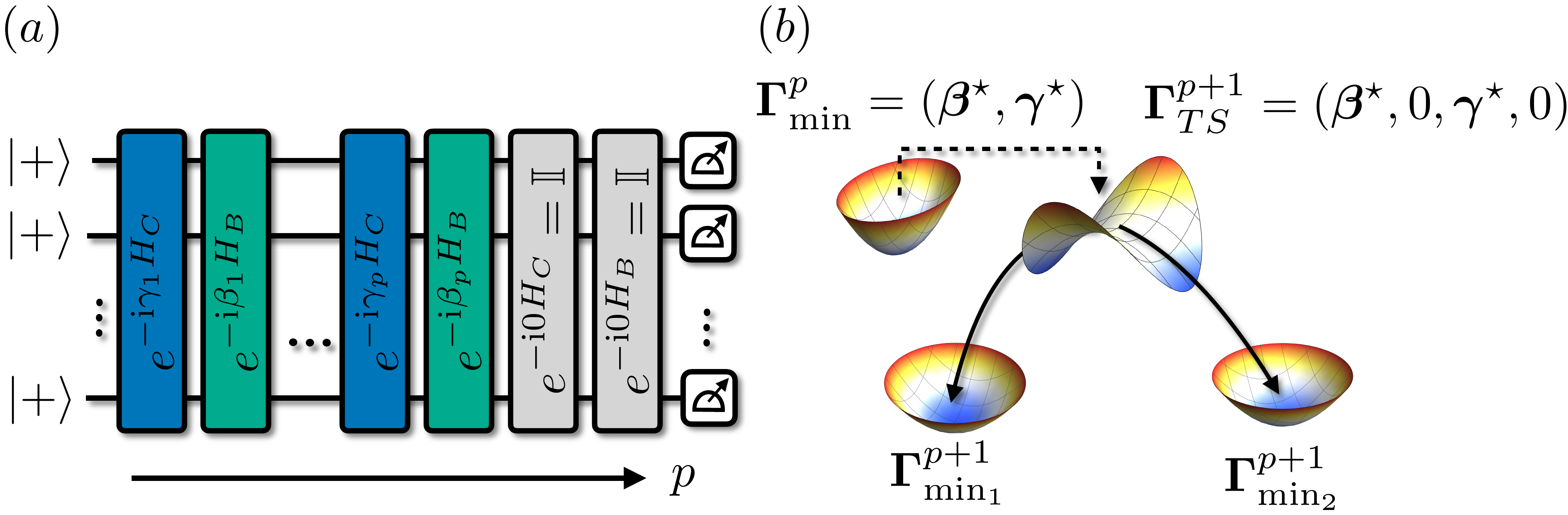}
    \caption{(a) Circuit diagram that implements the QAOA ansatz state with circuit depth $p$, see Eq.~(\ref{eq:QAOA_ansatz}). Gray boxes indicate the identity gates that are inserted when constructing a TS, as indicated in Theorem~\ref{th:ts}. (b) Local minima $\bm{\Gamma}_{\min}^p$ of QAOA$_p$ generate  a TS  $\bm{\Gamma}_{\text{TS}}^{p+1}$ for QAOA$_{p+1}$ that connects to two \emph{new local minima}, $\bm{\Gamma}_{\min_{1,2}}^{p+1}$ with lower energy.}
    \label{fig:1}
\end{figure}

\subsection{Energy minima and transition states}
Previous studies of the QAOA landscape were restricted to local minima of the cost function $E(\bm{\beta}, \bm{\gamma})$, since they can be directly obtained using standard gradient-based or gradient-free optimization routines. Local minima are stationary points of the energy landscape~(defined as $\partial_i E(\bm{\beta}, \bm{\gamma})=0$ for derivative running over all $i=1,\ldots,2p$ variational angles), where all eigenvalues of the Hessian matrix $H_{ij}=\partial_i \partial_j E(\bm{\beta}, \bm{\gamma})$ are positive, that is the Hessian at the local minimum is positive-definite. However, the study of energy landscapes~\cite{wales_2004} of chemical reactions and molecular dynamics has shown that TS, which corresponds to stationary points with a single negative eigenvalue of the Hessian matrix (index-1), also plays an important role~\footnote{Note, that on physical grounds we do not consider singular Hessians that have one or more vanishing eigenvalues, see Appendix~\ref{app:ts}.}. There, TS are particularly relevant as they correspond to the highest-energy configurations along a reaction pathway. They often serve as bottlenecks in the reaction process and thus are crucial for understanding reaction rates, designing catalysts, and predicting chemical behavior. By studying the role of transition states in the QAOA landscape, we aim to uncover insights that could lead to improved optimization strategies or better convergence properties of the algorithm. This motivates the construction of TS achieved below.

\subsection{Analytic construction of transition states}

The structure of the QAOA variational ansatz allows us to analytically construct the TS of  QAOA$_{p+1}$ using any local minima of QAOA$_p$: 

\begin{theorem}[TS construction, simplified version]\label{th:ts}
Assume that we found a local minimum of QAOA$_p$ denoted as $\bm{\Gamma}_{\min}^{p}=( \bm{\beta}^\star, \bm{\gamma}^\star)=( \beta_{1}^\star, \ldots, \beta_{p}^\star, \gamma_{1}^\star, \ldots, \gamma_p^\star).$ Padding the vector of variational angles with zeros at positions $i$ and $j$, results in
\begin{equation}\label{Eq:TS-construction-main}
\begin{split}
\bm{\Gamma}^{p+1}_{\text{\rm TS}}(i, j) =(
\beta_1^\star, ..., \beta^\star_{j-1}, &0, \beta^\star_{j}, ..., \beta_p^\star, \\
\gamma_1^\star, ..., \gamma^\star_{i-1}, &0, \gamma^\star_{i}, ..., \gamma_p^\star)
\end{split}
\end{equation}
being a TS for QAOA$_{p+1}$ when $j=i$ or $j=i+1$ and $\forall i\in[1,p]$, and also for $i=j=p+1$.
\end{theorem}
\begin{proof}
    The argument consists of two steps. First, by relating the first derivative over newly introduced parameters to derivatives over existing angles we show that Eq.~(\ref{Eq:TS-construction-main}) is a stationary point of QAOA$_{p+1}$. More specifically, we observe that the gradient components where the zero insertion is made satisfy the following relations
\begin{equation}
\begin{split}
    \partial_{\beta_l}|\bm{\beta}, \bm{\gamma}\rangle_{\big \vert \bm{\Gamma}^{p+1}_{\text{TS}}(l, l)} &= \partial_{\beta_{l-1}}|\bm{\beta}, \bm{\gamma}\rangle_{\big \vert \bm{\Gamma}^p_{\text{min}}}, \\
    \partial_{\beta_l}|\bm{\beta}, \bm{\gamma}\rangle_{\big \vert \bm{\Gamma}^{p+1}_{\text{TS}}(l, l+1)} &= \partial_{\beta_{l}}|\bm{\beta}, \bm{\gamma}\rangle_{\big \vert \bm{\Gamma}^p_{\text{min}}}, \\
    \partial_{\gamma_l}|\bm{\beta}, \bm{\gamma}\rangle_{\big \vert \bm{\Gamma}^{p+1}_{\text{TS}}(l, l)} &= \partial_{\gamma_{l}}|\bm{\beta}, \bm{\gamma}\rangle_{\big \vert \bm{\Gamma}^p_{\text{min}}}, \\
    \partial_{\gamma_{l+1}}|\bm{\beta}, \bm{\gamma}\rangle_{\big \vert \bm{\Gamma}^{p+1}_{\text{TS}}(l, l+1)} &= \partial_{\gamma_{l}}|\bm{\beta}, \bm{\gamma}\rangle_{\big \vert \bm{\Gamma}^p_{\text{min}}}. 
\end{split}
\label{eq:grad_TS}
\end{equation}
Since $\nabla E(\bm{\beta}, \bm{\gamma})_{\big \vert \bm{\Gamma}^p_{\text{min}}}=0$, it directly follows that the TS constructed using Theorem~\ref{th:ts} are also stationary points.  
In the second step, we show that the Hessian at the TS has a single negative eigenvalue. To this end in the Appendix~\ref{app:ts} we show that we can always write the Hessian at the TS in the following form
\begin{equation}
     \label{eq:Hessian-big}
H\big[\bm{\Gamma}^{p+1}_{\text{TS}}(l, k)\big]=
\begin{pmatrix}
   H(\bm{\Gamma}^p_{\text{min}}) & v(l,k) \\
  v^T(l,k) & h(l,k)
  \end{pmatrix},
\end{equation}
where $H(\bm{\Gamma}^p_{\text{min}}) \in \mathbb{R}^{2p \times 2p}$, $v(l,k) \in \mathbb{R}^{2p \times 2}$ and, $h(l,k) \in \mathbb{R}^{2 \times 2}$. Here, the largest block $H(\bm{\Gamma}^p_{\text{min}})$ corresponds to the old Hessian at the stationary point. The matrix $h(l,k)$ corresponds to the second derivatives of the energy with respect to new parameters that are initially set to zero, whereas matrix $v(l,k)$ represents the ``mixing'' terms, with one derivative taking over the old parameters and the second derivative corresponding to one of the new parameters, which are initialized at zero. By employing this representation of the Hessian at the TS, we utilize the eigenvalue interlacing theorem ([Ref.~\cite{matrix_book}, Theorem 4 on page 117] summarized in Theorem~\ref{th:eigInterlacing}) to establish that $H[\bm{\Gamma}^{p+1}_{\text{TS}}(l, k)]$ has at most two negative eigenvalues. Subsequently, we prove that the determinant of $H[\bm{\Gamma}^{p+1}_{\text{TS}}(l, k)]$ is negative for each of the $2p+1$ transition states, which implies the presence of only one negative eigenvalue (i.e., the index-1 direction). It is important to note that this result is independent of the choice of classical Hamiltonian, which is fixed to encode \textsc{MaxCut} in this work.
\end{proof}

The simplified theorem above ignores the possibility of vanishing eigenvalues of the Hessian, which can be ruled out only on physical grounds. This issue and complete proof of the theorem are discussed in Appendix~\ref{app:ts}. 

\section{From transition states to QAOA intialization}\label{sec:ts_initialization}
\subsection{Initialization graph}
For each local minimum of QAOA${_p}$, Theorem~\ref{th:ts} provides $p+1$ symmetric TS where zeros are padded at the same position, $i=j$, like in Fig.~\ref{fig:1}(a), and additionally $p$ non-symmetric TS with $j=i+1$, where zeros are padded in adjacent layers of the QAOA circuit. Fig.~\ref{fig:1}(b) shows how one can descend from a given TS along the positive and negative index-1 direction, finding two new local minima of QAOA$_{p+1}$ with lower energy. Thus Theorem~\ref{th:ts} provides us with a powerful tool to systematically explore the local minima in the QAOA in a recursive fashion.

\begin{figure}[tb]
    \centering
    \includegraphics[width=\columnwidth]{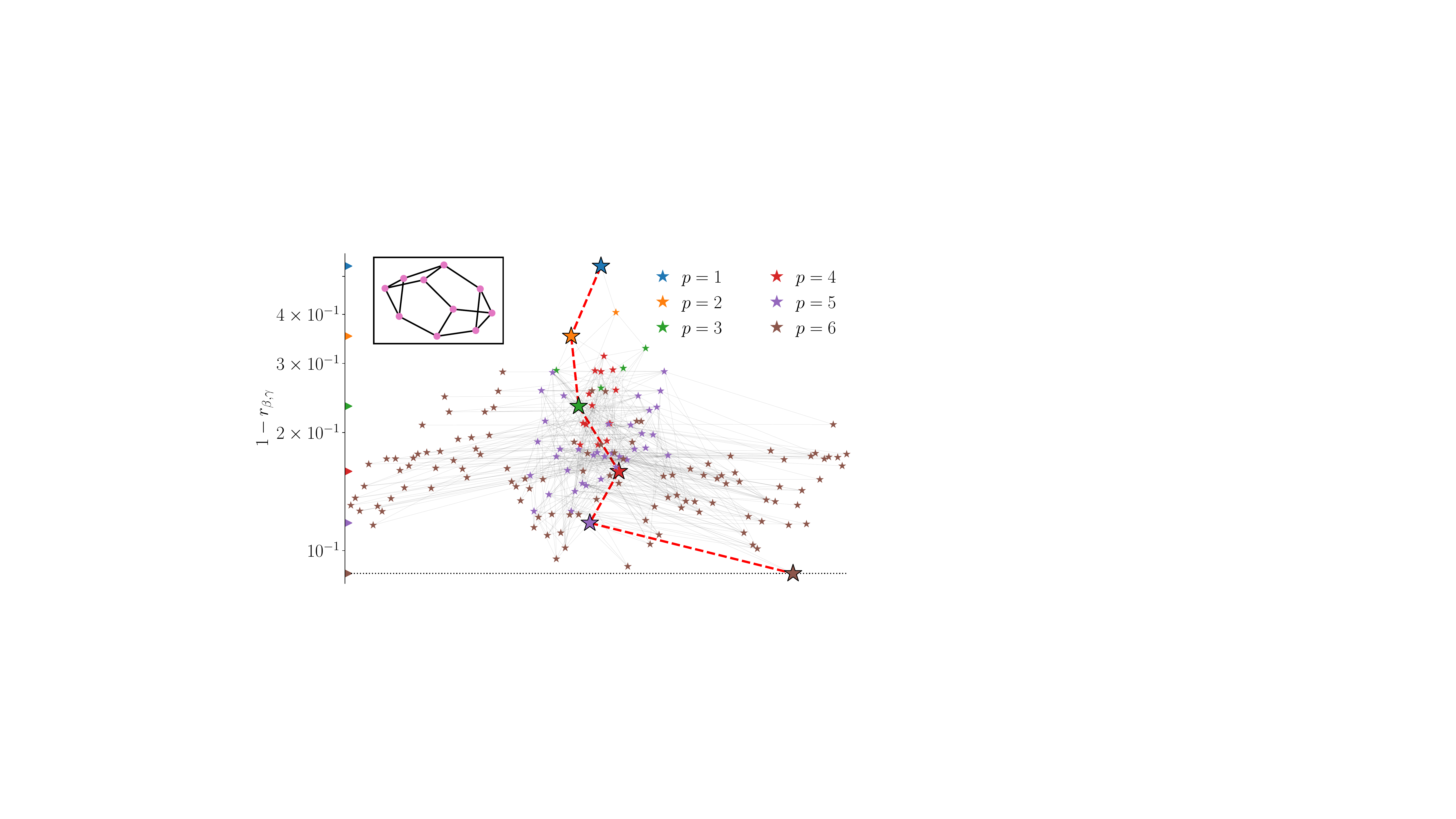}
    \caption{Initialization graph for the QAOA for \textsc{MaxCut} problem on a particular instance of RRG3 with $n=10$ vertices (inset). For each local minima of QAOA$_p$ we generate $p+1$ TS for QAOA$_{p+1}$, find corresponding minima as in Fig.~\ref{fig:1}(b), and show them on the plot connected by an edge to the original minima of QAOA$_{p+1}$. Position along the vertical axis quantifies the performance of QAOA via the approximation ratio, and points are displaced on the horizontal axis for clarity. Color encodes the depth of the QAOA circuit, and large symbols along with the red dashed line indicate the path that is taken by the \textsc{Greedy} procedure that keeps the best minima for any given $p$ resulting in an exponential improvement of the performance with $p$. The \textsc{Greedy} minimum coincides with an estimate of the global minimum for $p=6$ (dashed line) obtained by choosing the best minima from $2^p$ initializations on a regular grid.}
    \label{fig:2}
\end{figure}

Such exploration of the QAOA initializations for a particular graph with $n=10$ vertices is summarized in Fig.~\ref{fig:2}. We find a unique minimum for QAOA$_{1}$ using grid search (see Appendix~\ref{app:greedy}) in the fundamental region defined in Eq.~\ref{eq:symm} from which 
we construct two symmetric TS according to Eq.~(\ref{Eq:TS-construction-main}), descend from these TS in index-1 directions with the Broyden–Fletcher–Goldfarb–Shanno (BFGS)~\cite{bfgs1, bfgs2, bfgs3, bfgs4} algorithm, finding two new local minima of QAOA$_{2}$. These minima are connected to the minima of QAOA$_1$, since it was used to construct a TS. Repeating this procedure recursively for each of the $p+1$ symmetric TS~\footnote{Note, that we restrict only to symmetric TS since we numerically find no performance gain from including the non-symmetric TS in the initialization procedure.} we obtain the tree in Fig.~\ref{fig:2}. Assuming that all minima found in this way from symmetric TS are unique, their number would increase as $2^{p-1} p!$. Numerically, we observe that the number of unique minima is much smaller compared to the na\"ive counting, increasing approximately exponentially with $p$.

\subsection{\textsc{Greedy} maneuvering through the graph}

The exponential growth of the number of minima in QAOA depth $p$
makes the na\"ive construction and exploration of the full graph a challenging task. To deal with the rapidly growing number of minima we introduce:
\begin{corollary}[\textsc{Greedy} recursive strategy]
Using the lowest energy minimum that is found for QAOA depth $p$, we generate $2p+1$ transition states (TS) for QAOA$_{p+1}$. Each transition state corresponds to the same state in the Hilbert space as the initial local minimum, so the energy of all the transition states is the same and equal to the energy of the initial local minimum. We then optimize the QAOA parameters starting from each of these transition states and select the best new local minimum of QAOA$_{p+1}$ to iterate this procedure. This \textsc{Greedy} recursive strategy is guaranteed to lower energy at every step.
\end{corollary}
\begin{proof}
Let the initial local minimum at QAOA depth $p$ have energy $E_p$. Since all the $2p+1$ transition states are generated from this minimum and have the same energy $E_p$, when we optimize the QAOA parameters for QAOA$_{p+1}$ starting from these transition states, all the converged local minima will have energy less than or equal to $E_p$. As a result, the energy can either decrease or stay the same (provided that curvature vanishes, which we do not expect on physical grounds, see Appendix~\ref{app:ts}), but it cannot increase. Therefore, the \textsc{Greedy} recursive strategy is guaranteed to lower or maintain the energy at every step.
\end{proof}
The \textsc{Greedy} path that is taken by this strategy in the initialization graph is shown in Fig.~\ref{fig:2} as a red dashed line. We can see that this heuristic allows to very effectively maneuver the increasingly complex graph with its numerous local minima and find the global minimum for circuit depths up to $p=7$. A detailed description of the algorithm is presented in Appendix~\ref{app:greedy}.
\begin{figure}[b]
    \centering
    \includegraphics[width=\columnwidth]{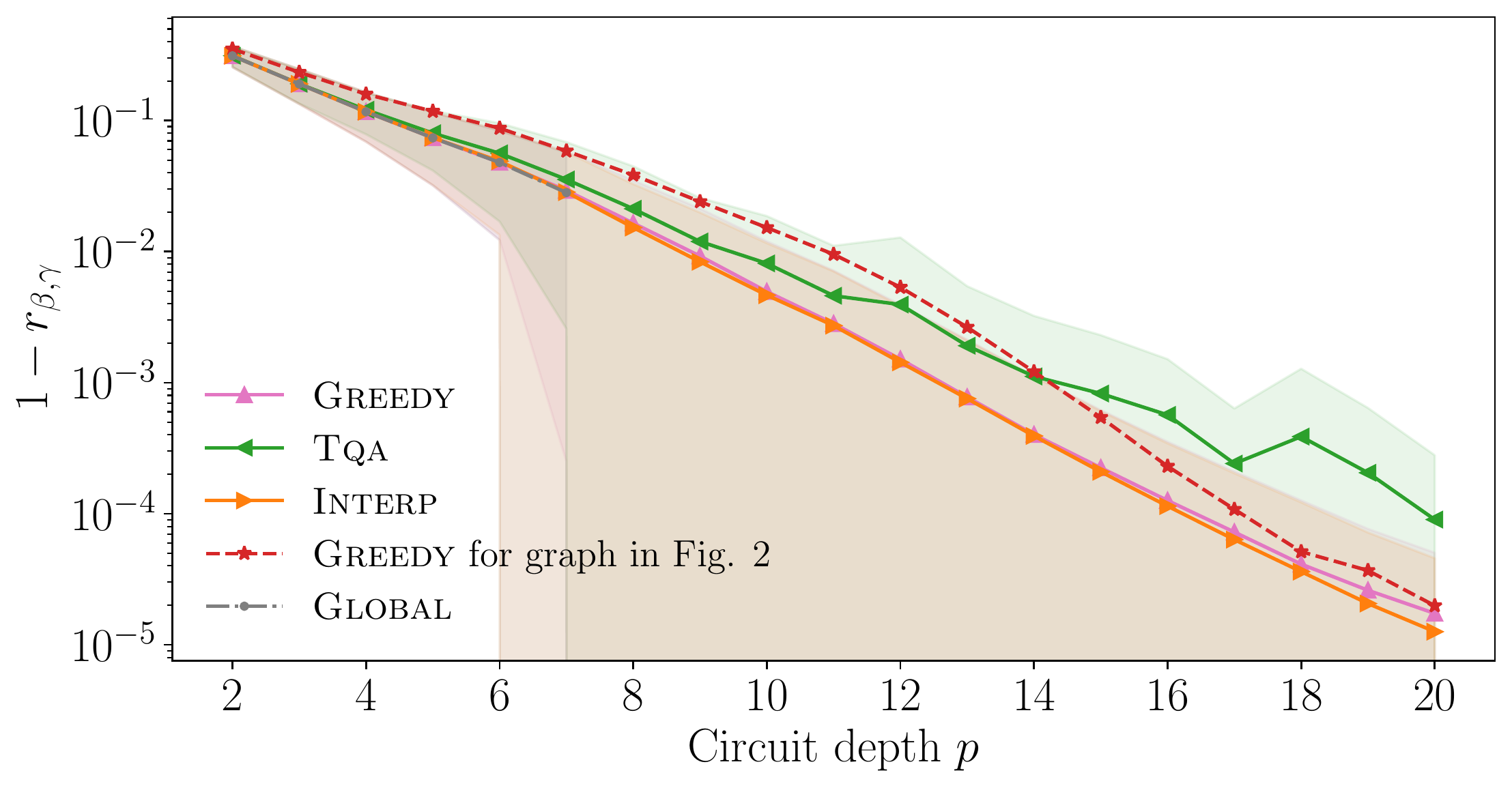}
    \caption{Performance comparison between different QAOA initialization strategies used for avoiding low-quality local minima. \textsc{Greedy} approach proposed in this work yields the same performance as \textsc{Interp}~\cite{zhou2018quantum} and slightly outperforms \textsc{TQA}~\cite{sack2021quantum} at large $p$.
    \textsc{Global} refers to the best minima found out of $2^p$ initializations on a regular grid. Data is averaged over 19 non-isomorphic RRG3 with  $n=10$, shading indicates standard deviation. System size scaling for up to $n=16$ and performance comparison for different graph ensembles can be found in the Appendix~\ref{App:last}.}
    \label{fig:3}
\end{figure}

To systematically explore how \textsc{Greedy} maneuvers the initialization graph, we compare it to two initialization strategies proposed in the literature: The so-called \textsc{Interp} approach~\cite{zhou2018quantum} interpolates the optimal parameters found for circuit depth $p$ to $p+1$ and uses it as a subsequent initialization. This procedure creates a \textit{smooth parameter pattern} that mimics an annealing schedule. Numerical studies demonstrated that \textsc{Interp} has the same performance as the best out of $2^p$ random initializations. The second method that we use for comparison is the Trotterized quantum annealing (\textsc{TQA}) method~\cite{sack2021quantum}, that initializes QAOA$_p$ using $\gamma_j=(1-\frac{j}{p})\Delta t$ and  $\beta_j=\frac{j}{p}\Delta t$. The step size $\Delta t$ is a free parameter determined in a pre-optimization step.
The \textsc{TQA} has similar performance to \textsc{Interp} at moderate circuit depths, notably having lower computational cost. Obtaining an initialization for QAOA$_p$ within the \textsc{Interp} framework requires running the optimization for all $p'=1,\ldots, p-1$, while in the \textsc{TQA} the search for an optimal $\Delta t$ is performed directly for a given $p$.

Fig.~\ref{fig:3} reveals that the \textsc{Greedy} approach yields similar performance to existing methods. Moreover, the performance of \textsc{TQA} slightly degrades at higher $p$, however, \textsc{Greedy} is fully on par with \textsc{Interp} initialization. The comparable performance between \textsc{Greedy} and earlier heuristic approaches is surprising. Indeed, the \textsc{Greedy} method for QAOA$_{p}$ explores $p+1$ symmetric TSs and chooses the best out of the resulting up to $2(p+1)$ minima (if none are equivalent), in contrast to \textsc{Interp}, which uses a single smooth initialization pattern at every $p$ and thus at a smaller computational cost. 

\subsection{Smooth pattern of variational angles and heuristic initializations}

\begin{figure*}[t]
    \centering
    \includegraphics[width=0.75\textwidth]{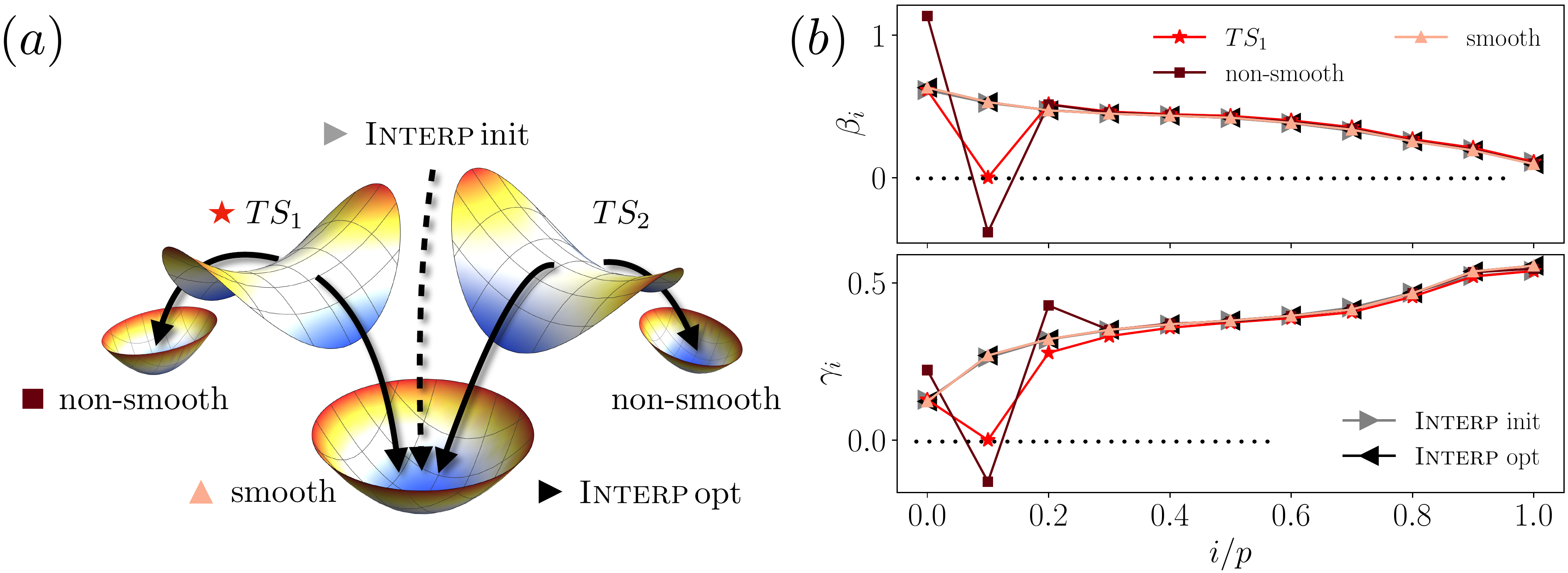}
    \caption{
    \label{fig:4}
    (a) Cartoon of descent from two different TS at of QAOA$_{p+1}$ generated from a QAOA$_p$ minimum with a smooth pattern leads to the same new smooth pattern minima of QAOA$_{p+1}$, also reached from the \textsc{Interp}~\cite{zhou2018quantum} initialization. Two additional non-smooth local minima typically have higher energy. (b) shows the corresponding initial and convergent parameter patterns for the RRG3 graph shown in Fig.~\ref{fig:2} for $p=10$. 
    }
\end{figure*}

We find that having a smooth dependence of the variational angles on $p$ (referred to as a ``smooth pattern'') is an important characteristic for efficiently maneuvering the initialization graph. A smooth pattern means that the variational angles change gradually and continuously as the QAOA depth $p$ increases, without abrupt jumps or discontinuities. This smoothness property can be visually inspected by plotting the variational angles as a function of $p$ and observing whether the curve appears continuous and smooth. Assuming we found a smooth pattern of QAOA$_{p}$, Theorem~\ref{th:ts} produces a TS of QAOA$_{p+1}$ by padding it with zeros, effectively introducing a discontinuity (bump).
Optimization from the TS with such a bump can proceed by rolling down either side of the saddle, see Fig.~\ref{fig:4}(a), finding two new minima. Remarkably, the eigenvector corresponding to the index-1 direction of the Hessian has dominant weight on the variational angles with initially zero value, see~\ref{app:index_1} for details.
Thus descending along the index-1 direction, we can either enhance or heal the resulting discontinuity in the pattern of variational angles. 
As a result, among two new local minima of QAOA$_{p+1}$ one typically exhibits a smooth parameter pattern where the bump was removed, while the other minimum has an enhanced discontinuity, see Fig.~\ref{fig:4}(b) for an example. 
Utilizing these observations in a numerical study, we find that minima exhibiting a non-smooth parameter pattern exhibit usually a worse or the same performance as smooth minima. 
In fact, in the \textsc{Greedy} procedure we find that in most cases, in particular in the beginning of the protocol, smooth minima are selected. However, there are cases where a non-smooth minimum is selected if it exhibits the same energy as the smooth one. 
\textsc{Greedy} then branches off in the optimization graph into a sub-graph involving only non-smooth minima. Usually, this process of branching off is followed by a smaller gain in performance from increasing $p$.

The preferred smoothness of QAOA optimization parameters has been explored in the literature~\cite{zhou2018quantum, mele2022avoiding, wurtz2022counterdiabaticity} and is believed to be linked to quantum annealing~\cite{brady2021optimal} (QA). In QA the ground state of the Hamiltonian $H_C$ is obtained by preparing the ground state of $H_B$ and smoothly evolving the system to $H_C$ such that the system remains in the ground state during the evolution. A fast change, as generated by a bump in the protocol, leads to leakage into excited energy levels and thus decreased overlap with the target ground state of $H_C$. Since the QAOA can be understood as a Trotterized version of QA~\cite{farhi2014quantum, sack2021quantum, zhou2018quantum}, for large $p$, we believe that a similar process is present in the QAOA and thus makes a smooth parameter pattern preferable.

We find that smooth \textsc{Greedy} minima coincide with \textsc{Interp} minima as shown in Fig.~\ref{fig:4}(b). The \textsc{Interp} naturally creates a smooth parameter pattern since the minima found at $p$ is interpolated to a QAOA$_{p+1}$ initialization. The optimizer only slightly alters the parameters from its initial value, as can be seen in Fig.~\ref{fig:4}(b). 
Geometrically, the \textsc{Interp} initialization can be obtained from the symmetric TS constructed by Theorem~\ref{th:ts} as $\bm{\Gamma}^{p+1}_{\text{\textsc{Interp}}}=\frac{1}{p}\sum_{i=1}^{p+1} \bm{\Gamma}^{p+1}_{\text{TS}}(i,i)$. In other words, $\bm{\Gamma}^{p+1}_{\text{\textsc{Interp}}}$ is the \emph{rescaled center of mass point} of all symmetric TS, with the rescaling factor $(p+1)/p$  being physically motivated. Considering the center of mass of all TS smoothens out discontinuities present in individual TS. The re-scaling is related to the notion of ``total time'' of the QAOA, given by the sum of all variational angles, $T=\sum_j |\gamma_j|+|\beta_j|$~\cite{zhou2018quantum, liang2020investigating}, that resembles the total annealing time in the limit $p\to\infty$. This parameter has been shown to scale as $T\sim p$~\cite{sack2021quantum}, naturally explaining the role of factor $(p+1)/p$ in yielding the correct increased total time of QAOA$_{p+1}$.
In other words, the \textsc{Interp} strategy seems to essentially execute a \textsc{Greedy} search without optimizing in the index-1 direction from the TS. This insight lends credence to the success of \textsc{Interp}. However, only \textsc{Greedy} offers a guarantee for performance improvement with increasing $p$, while for \textsc{Interp} this behavior is supported only by numerical simulations.

\section{Discussion}\label{sec:discussion}

In this work we analytically demonstrated that minima of QAOA$_p$ can be used to obtain transition states (TS) for QAOA$_{p+1}$ which are stationary points with a unique negative eigenvalue in the Hessian. 
These TS provide an excellent initialization for QAOA$_{p+1}$, because they connect to two new local minima with lower energy. This construction allows us to visualize how local minima emerge at different energies for increasing circuit depth using an initialization graph. Categorizing the local minima on this graph by their smooth (discontinuous) patterns of variational parameters, we find that the smooth minima achieve the best performance. Incorporating the smooth nature of minima allows us to  establish a relation between the \textsc{Greedy} approach for the exploration of the initialization graph and the best available initialization strategy~\cite{zhou2018quantum}. 

The use of TS and their analytic construction for the study of QAOA provide the first steps towards an in-depth understanding of the full optimization landscape of the QAOA. 
The constructed TS are guaranteed to provide an initialization that improves the QAOA performance, suggesting that our construction may be useful for establishing analytic QAOA performance guarantees~\cite{farhi2014quantum,wurtz2021maxcut, farhi2020needs} for large $p$ in a recursive fashion. Of particular interest is here an analytical understanding of the numerically observed exponential performance improvement with circuit depth. On a practical side, the established relation between heuristic initializations~\cite{zhou2018quantum} and \textsc{Greedy} exploration of TS suggests that our construction of TS may be useful as a starting point for constructing simple initialization strategies in a broader class of quantum variational algorithms, such as the variational quantum eigensolver~\cite{kandala2017hardware-efficient, peruzzo2014vqe} and quantum machine learning~\cite{benedetti2019qml}. 

In addition, our results invite a more complete characterization of the QAOA landscape using the energy landscapes perspective~\cite{wales_2004}. What fraction of minima does our procedure find out of the complete set of QAOA local minima? Are there more TS and are our analytically constructed TS typical? How is the Hessian spectrum distributed at these minima and TS? 
How do these properties depend on the choice of the QAOA classical Hamiltonian, particularly for classical problems with intrinsically hard landscapes~\cite{chou2021limitations}? 
Answering these and related questions will most likely lead to practical ways of further speeding up the QAOA by reducing the overhead of the classical optimization~\cite{weidenfeller2022scaling}. 

\begin{acknowledgements}
We thank V.~Verteletskyi for a joint collaboration on numerical studies of the QAOA  during his internship at ISTA that inspired analytic results on TS reported in this work.
We acknowledge A.~A.~Mele and M.~Brooks for discussions and D.~Egger, P.~Love, and D.~Wierichs for valuable feedback on the manuscript. S.H.S., R.A.M., and M.S.\ acknowledge support by the European Research Council (ERC) under the European Union's Horizon 2020 research and innovation program (Grant Agreement No.~850899).
R.K.\ is supported by the SFB BeyondC (Grant No.~F7107-N38) and the project QuantumReady (FFG 896217).
\end{acknowledgements}
\appendix
\section{Restricting QAOA parameter space by symmetries}\label{appx:symmetries}
In this Appendix, we find the symmetry properties of the cost function 
$$
E(\bm \beta, \bm \gamma)=\langle \bm \beta, \bm \gamma| H_C | \bm \beta, \bm \gamma \rangle
$$ 
for the QAOA$_p$ (i.e.\ QAOA with circuit depth $p$) ansatz. Here we use bold notation for both $\beta$ and $\gamma$ parameters to denote a length-$p$ vector of angles, i.e.\ $\bm{\beta}=(\beta_1,\ldots,\beta_p)$ and $\bm{\gamma}=(\gamma_1,\ldots,\gamma_p)$. The use of symmetries allows to restrict the manifold of variational parameters, leading to a more efficient exploration of the QAOA landscape. This section expands upon previous results by~\cite{zhou2018quantum}.

We begin by rewriting the exponents of both classical and mixing Hamiltonian as:
\begin{align}
	e^{-\mathrm{i}\beta_l H_B}&=\prod_{k=1}^n e^{-\mathrm{i}\beta_l\sigma_k^x}=(\cos{\beta_l}-\mathrm{i}\sin{\beta_l}\, \sigma^x)^{\otimes n},\\
	e^{-\mathrm{i}\gamma_l H_C}&=\prod_{\langle j,k\rangle}e^{-\mathrm{i}\gamma_l \sigma_j^z\sigma_k^z}=\prod_{\langle j,k\rangle}(\cos{\gamma_l}-\mathrm{i}\sin{\gamma_l}\sigma_j^z\sigma_k^z).
\end{align}
From here it is apparent that adding $\pi$ to any of the parameters, $\beta_l, \gamma_l \to \beta_l+\pi, \gamma_l+\pi$ for all $l \in [{1,p}]$ does not change the cost function value $E(\bm \beta, \bm \gamma)$. Indeed, this leads to an appearance of an overall negative sign that cancels within the expectation value of the classical Hamiltonian. Therefore we can easily restrict the search space to  {\bf (i)} $\beta_l, \gamma_l \in[-\frac{\pi}{2},\frac{\pi}{2}]$.

For $\beta$ parameters we can restrict the parameter space even further. In Ref.~\cite{zhou2018quantum} the authors restrict the parameters as $\beta_l \in[-\frac{\pi}{4},\frac{\pi}{4}]$ due to the following considerations. Consider adding $\frac{\pi}{2}$ to $\beta$, the exponent $e^{-\mathrm{i}\left(\beta_l+\frac{\pi}{2}\right) H_B}=e^{-\mathrm{i}\beta_l H_B}e^{-\mathrm{i}\frac{\pi}{2}H_B}$ leads to an additional product of all $\sigma^x$ operators,
\begin{equation}
	e^{-\mathrm{i}\frac{\pi}{2}H_B}=(-\mathrm{i}\sigma^x)^{\otimes n}.
\end{equation}
this operator flips all spins, effectively being a generator of the $Z_2$ symmetry of the classical Ising Hamiltonian, $H_C$. Therefore, such a shift of $\beta_l$ will have no effect on the cost function and we restrict {\bf (ii)} $\beta_l \in[-\frac{\pi}{4},\frac{\pi}{4}]$.

Yet another symmetry is recovered by taking the complex conjugate of the energy. As both classical and mixing Hamiltonians are real-valued, one has
\begin{align}
	E^*(\bm{\beta}, \bm{\gamma})&=\langle \bm \beta, \bm \gamma| H_C |\bm \beta, \bm \gamma \rangle^* = E(-\bm{\beta}, -\bm{\gamma}).
\end{align}
And because the energy is also real-valued ($H_C$ is Hermitian), we recover another symmetry of the cost function: {\bf (iii)} $(\bm{\beta},\bm{\gamma})\to(-\bm{\beta},-\bm{\gamma})$. 

The symmetries {\bf{(i)}}-{\bf{(iii)}} introduced above were discussed in Refs.~\cite{zhou2018quantum, sack2021quantum}. But we can restrict the search space even further. In particular, we demonstrate that for the QAOA cost function for 3-regular random graphs (RRG3) the following \emph{additional} symmetry holds:
\begin{enumerate}
	\item[\bf (iv)] Flipping sign of any of the $\beta_l\to-\beta_l$ for any $l\in [1,p]$ together with shifts of $\gamma_{l,l+1}$ angles, as  $\gamma_{l,l+1}\to\gamma_{l,l+1}\pm\frac{\pi}{2}$. Note that for $l=p$ only the $\gamma_p$ angle has to be shifted.
\end{enumerate}
Let us prove this property for regular graphs with odd connectivity (i.e. 3-regular, 5-regular, \ldots).  In order to demonstrate the property {\bf{(iv)}} for $j<p$, it is enough to show that:
\begin{equation}
	e^{-\mathrm{i}\frac{\pi}{2}H_C}e^{\mathrm{i}\beta H_B}e^{-\mathrm{i}\frac{\pi}{2}H_C}\sim e^{-\mathrm{i}\beta H_B},
\end{equation}
where $\sim$ stands for equivalence up to a global phase. In other words, we use the property that $e^{-\mathrm{i}\frac{\pi}{2}H_C}\sim \prod_i \sigma^z_i$ acts as a product of $\sigma^z$ operators over all spins, that relies on the fact that each vertex is connected to an odd number of edges (interaction terms). This leads to the relation
\begin{equation}
	e^{-\mathrm{i}\frac{\pi}{2}H_C}e^{\mathrm{i}\beta H_B}e^{-\mathrm{i}\frac{\pi}{2}H_C}
	\sim e^{-\mathrm{i}\beta H_B}.
\end{equation}
Thus, the change of sign of $\beta_k$ can be compensated by the shifts of ``adjacent'' angles $\gamma_{k,k+1}$ by $\pi/2$, leading to the property {\bf{(iv)}} when $j<p$. In the particular case of $j=p$, the property {\bf{(iv)}} for $j=p$ is obtained using the following relation 
\begin{align}
	&e^{\mathrm{i}\frac{\pi}{2}H_C}e^{-\mathrm{i}\beta H_B}H_Ce^{\mathrm{i}\beta H_B}e^{-\mathrm{i}\frac{\pi}{2}H_C} \\
	\sim & e^{\mathrm{i}\frac{\pi}{2}H_C}e^{-\mathrm{i}\beta H_B}e^{\mathrm{i}\frac{\pi}{2}H_C}H_Ce^{-\mathrm{i}\frac{\pi}{2}H_C}e^{\mathrm{i}\beta H_B}e^{-\mathrm{i}\frac{\pi}{2}H_C}\\
	= & e^{\mathrm{i}\beta H_B}H_Ce^{-\mathrm{i}\beta H_B}.
\end{align}
 
Finally, let us rewrite the property {\bf{(iv)}} by sequentially applying this symmetry for all indices $j$ starting from $k$ and ending at $p$. Then we obtain the following property equivalent to {\bf{(iv)}} and dubbed {\bf{(iv')}}:
\begin{enumerate}
\item[{\bf (iv')}]$	\forall j=[k,p]
: \ \beta_j\to-\beta_j, \ \gamma_j\to\gamma_j\pm\frac{\pi}{2}$.
\end{enumerate}
This allows us to restrict all $\gamma$ angles to the region $[-\frac{\pi}{4},\frac{\pi}{4}]$. Moreover, the sign-flip symmetry {\bf{(iii)}} allows us to make one of the $\gamma$ angles, for instance, $\gamma_1$, positive, cutting the search space in half.

In addition, let us apply property {\bf{(iv')}} for $k=1$ (i.e. including all layers of the unitary circuit) and supplement it with a global sign flip, operation {\bf{(iii)}}. As a result, we obtain the following symmetry:
\begin{equation}
	\gamma_1\to\pm\frac{\pi}{2}-\gamma_1,\ \forall j=[2,p]: \ \gamma_j\to-\gamma_j
\end{equation}
This indicates that there is a $p$-dimensional plane in the landscape with coordinates $\bm{\gamma}=(\pm\frac{\pi}{4},\bm{0}_{p-1})$ which acts as a mirror. This plane is characterized by a vanishing gradient of the cost function and the Hessian having $p$ vanishing eigenvalues. However, it is located on the edge of our search space and it has a vanishing expectation value of the cost function, corresponding to the approximation ratio $r=0$, which is very far from the good-quality local minima. 

In summary, collecting all symmetries discussed above, we restrict the fundamental search region to 
\begin{align}
 \beta_l &\in \bigg[-\frac{\pi}{4},\frac{\pi}{4}\bigg], \; \forall l \in [1,p], \\
 0 &< \gamma_1 < \frac{\pi}{4}, \\
\gamma_j &\in \bigg[-\frac{\pi}{4},\frac{\pi}{4}\bigg], \; \forall j \in [2,p].
\end{align}
\section{Construction of transition states} \label{app:ts}
In this section, we show how to use a local minimum of the QAOA$_p$  to construct a set of $2p+1$ transition states (TS) at circuit depth $p+1$. These are stationary points with all but one Hessian eigenvalue being positive. 
More precisely, we show the following statement:

\begin{theorem}[TS construction, full version]
Let $\bm{\Gamma}_{\min}^{p}=( \bm{\beta}^\star, \bm{\gamma}^\star)=( \beta_{1}^\star, \ldots, \beta_{p}^\star, \gamma_{1}^\star, \ldots, \gamma_p^\star)$ be a local minimum of QAOA$_p$. Define the following $2p+1$ points by padding this vector with zeroes at distinguished positions:
\begin{equation}\label{Eq:TS-construction}
\begin{split}
\bm{\Gamma}^{p+1}_{\text{\rm TS}}(i, j) =(
\beta_1^\star, ..., \beta^\star_{j-1}, &0, \beta^\star_{j}, ..., \beta_p^\star, \\
\gamma_1^\star, ..., \gamma^\star_{i-1}, &0, \gamma^\star_{i}, ..., \gamma_p^\star)
\end{split}
\end{equation}
with $i \in \left[1,p+1\right]$ and $j=i$ or $j=i+1$. Then each of these points is either (i) a TS for QAOA$_{p+1}$ or (ii) has a non-regular Hessian.
\end{theorem}

Theorem~\ref{th:ts} in the main text is a streamlined version of this statement that does not mention the possibility of degenerate Hessians. We expect that the Hessian matrix of a local minimum of QAOA$_p$ is non-degenerate in the absence of symmetries and provided the circuit is not overparametrized~\cite{larocca2021theory} (if there exists some combination of variational angles, such that its changes do not influence the quantum state, it leads to vanishing eigenvalue of Hessian). Analogously, in the case of the Hessian at the TS of QAOA$_{p+1}$, we numerically find that option (ii) never happens. Below, we relate the two new additional eigenvalues of the Hessian at the TS to the expectation value of a physical operator over the variational state. This expectation value is non-zero in the absence of special symmetries or fine-tuning, providing a physical justification for why we do not observe zero eigenvalues in the Hessian spectra of our TS.

\subsection{Cost function gradient}
Let us start by computing the energy gradient $\nabla E(\bm{\beta}, \bm{\gamma})$. Derivatives of the quantum state with respect to parameters $\beta_l, \gamma_l$ are given by the following expressions:
\begin{equation}
\begin{split}
    \partial_{\beta_l} |\bm{\beta}, \bm{\gamma}\rangle &= -\mathrm{i} U_{>l}H_B U_{\leq l}|+\rangle, \\
    \partial_{\gamma_l} |\bm{\beta}, \bm{\gamma}\rangle &= -\mathrm{i} U_{\geq l} H_C U_{< l}|+\rangle,
\end{split}
\label{eq:grad_state}
\end{equation}
where $U_{\geq l} = U_B(\beta_p)U_C(\gamma_p)\cdots U_B(\beta_l)U_C(\gamma_l)$, $U_{\leq l}=U_B(\beta_l)U_C(\gamma_l)\cdots U_B(\beta_1)U_C(\gamma_1)$ and analogously for $U_{< l}$, and $U_{>l}$. For simplified notation we use write $\ket{+}$ instead of $\ket{+}^{\otimes n}$. 
We can now deduce the components of the  energy gradient $\nabla E(\bm{\beta}, \bm{\gamma})$ from Eq.~\eqref{eq:grad_state}. They read
\begin{equation}
\begin{split}
    \partial_{\beta_l} E(\bm{\beta}, \bm{\gamma}) &= \mathrm{i} \langle +| U^\dagger_{ \leq l} [  H_B,  U^\dagger_{>l} H_C U_{>l} ] U_{\leq l} |+\rangle,\\
    \partial_{\gamma_l} E(\bm{\beta}, \bm{\gamma}) &= \mathrm{i} \langle +| U^\dagger_{ < l} [  H_C,  U^\dagger_{\geq l} H_C U_{\geq l} ] U_{<l} |+\rangle.
\end{split}
\label{eq:grad_energy}
\end{equation}

Our goal is to prove that given a local minimum $\bm{\Gamma}^p_{\text{min}}  = (\beta_1^\star,\ldots,\beta_p^\star, \gamma_1^\star,\ldots,\gamma_p^\star)$ for a QAOA$_p$ the set of $2p+1$ points 
\begin{equation}
\label{eq:ts2}
\begin{split}
\bm{\Gamma}^{p+1}_{\text{TS}}(l, k) =(\beta_1^\star, ..., \beta^\star_{l-1}, 0, \beta^\star_{l}, ..., \beta_p^\star,
\\
\gamma_1^\star, ..., \gamma^\star_{k-1}, 0, \gamma^\star_{k}, ..., \gamma_p^\star
),
\end{split}
\end{equation}
with $l$ ranging from $1$ to $p+1$ and either $k=l$ or $k=l+1$ are all TSs. 
The first step is to prove that they are all stationary points. That is, each such point leads to a vanishing gradient. From the above expression, it follows that we only have to consider gradient components where the zero insertion is made since the others are zero due to the point $\bm{\Gamma}^p_{\text{min}}$ being a local minimum (i.e. derivatives are vanishing). For the derivatives over newly introduced angles using Eq.~\eqref{eq:grad_state}, we see that
\begin{equation}
\begin{split}
    \partial_{\beta_l}|\bm{\beta}, \bm{\gamma}\rangle_{\big \vert \bm{\Gamma}^{p+1}_{\text{TS}}(l, l)} &= \partial_{\beta_{l-1}}|\bm{\beta}, \bm{\gamma}\rangle_{\big \vert \bm{\Gamma}^p_{\text{min}}}, \\
    \partial_{\beta_l}|\bm{\beta}, \bm{\gamma}\rangle_{\big \vert \bm{\Gamma}^{p+1}_{\text{TS}}(l, l+1)} &= \partial_{\beta_{l}}|\bm{\beta}, \bm{\gamma}\rangle_{\big \vert \bm{\Gamma}^p_{\text{min}}}, \\
    \partial_{\gamma_l}|\bm{\beta}, \bm{\gamma}\rangle_{\big \vert \bm{\Gamma}^{p+1}_{\text{TS}}(l, l)} &= \partial_{\gamma_{l}}|\bm{\beta}, \bm{\gamma}\rangle_{\big \vert \bm{\Gamma}^p_{\text{min}}}, \\
    \partial_{\gamma_{l+1}}|\bm{\beta}, \bm{\gamma}\rangle_{\big \vert \bm{\Gamma}^{p+1}_{\text{TS}}(l, l+1)} &= \partial_{\gamma_{l}}|\bm{\beta}, \bm{\gamma}\rangle_{\big \vert \bm{\Gamma}^p_{\text{min}}}, 
\end{split}
\label{eq:grad_TS}
\end{equation}
where the index $l$ ranges from $1$ to $p+1$ for the $(l,l)$ case and from $1$ to $p$ in the $(l,l+1)$ case. 

These observations reduce the derivatives over the new angles to derivatives over angles from local minima of QAOA$_{p}$. And these vanish by definition because we started in a local minimum which is itself a stationary point, that is 
\begin{equation}
    \nabla E(\bm{\beta}, \bm{\gamma})_{\big \vert \bm{\Gamma}^p_{\text{min}}}=0.
\end{equation}
We emphasize that these arguments do not apply to two special cases that should be treated separately. 

In particular, Eq.~\eqref{eq:grad_state} does not provide any information for: {\bf{(i)}} the gradient component $\partial_{\beta_1}[\cdot]$ when considering TS $\bm{\Gamma}^{p+1}_{\text{TS}}(1, 1)$ and $\bm{\Gamma}^{p+1}_{\text{TS}}(1, 2)$, and 
{\bf{(ii)}} the gradient component $\partial_{\gamma_{p+1}}[\cdot]$ when considering points $\bm{\Gamma}^{p+1}_{\text{TS}}(p+1, p+1)$. For case  {\bf{(i)}}, we use that $H_B|+\rangle = n|+\rangle$ with $n$ being the number of qubits, to show that
\begin{align}
\partial_{\beta_1}|\bm{\beta}, \bm{\gamma}\rangle_{\big \vert \bm{\Gamma}^{p+1}_{\text{TS}}(1, k)} =& -\mathrm{i} n |\bm{\beta}, \bm{\gamma}\rangle_{\big \vert \bm{\Gamma}^{p}_{\text{min}}}
\end{align}
for $k=1,2$. This in turn implies
\begin{align}
\partial_{\beta_{1}} E(\bm{\beta}, \bm{\gamma})_{\big \vert \bm{\Gamma}^{1}_{\text{TS}}(1, k)} =& (\mathrm{i} n - \mathrm{i} n) \langle \bm{\beta}, \bm{\gamma} | \bm{\beta}, \bm{\gamma} \rangle_{\big \vert \bm{\Gamma}^{p}_{\text{min}}} = 0, 
\end{align}
as desired. For case {\bf{(ii)}} we have that
\begin{align}
    \partial_{\gamma_{p+1}} |\bm{\beta}, \bm{\gamma}\rangle_{\big \vert \bm{\Gamma}^{p+1}_{\text{TS}}(p+1, p+1)} = - \mathrm{i} H_C | \bm{\beta}, \bm{\gamma}\rangle_{\big \vert \bm{\Gamma}^{p}_{\text{min}}},
    \end{align}
which handles the second special case:
\begin{align}
    \partial_{\gamma_{p+1}} E(\bm{\beta}, \bm{\gamma})_{\big \vert \bm{\Gamma}^{p+1}_{\text{TS}}(p+1, p+1)} = (\mathrm{i}-\mathrm{i}) E(\bm{\Gamma^p_{\text{min}}}) = 0.   
\end{align}
Putting everything together implies that all energy partial derivatives vanish for every $\bm{\Gamma}^{p+1}_{\mathrm{TS}}$ introduced in Theorem~\ref{th:ts}:
\begin{equation}
\nabla E (\bm{\beta},\bm{\gamma})_{\big \vert \bm{\Gamma}^{p+1}_{\text{TS}}(l,l)}= \nabla E (\bm{\beta},\bm{\gamma})_{\big \vert \bm{\Gamma}^{p+1}_{\text{TS}}(l,l+1)}=0
\end{equation}
for all $l \in [1,p+1]$ except the pair $(p+1,p+2)$ which exceeds the index range. 
In other words: these $2(p+1)-1=2p+1$ points must all be stationary points.

\subsection{Cost function Hessian}
We now proceed with the study of the Hessian for each of the stationary states in the set $\bm{\Gamma}^{p+1}_{\text{TS}}(l, k)$ with $l$ ranging from 1 to $p+1$ and $k$ being $l$ or $l+1$. Using basic row and column operations we  decompose the Hessian as follows:
\begin{equation}
     \label{eq:Hessian-big}
H[\bm{\Gamma}^{p+1}_{\text{TS}}(l, k)]=
\begin{pmatrix}
   H(\bm{\Gamma}^p_{\text{min}}) & v(l,k) \\
  v^T(l,k) & h(l,k)
  \end{pmatrix},
\end{equation}
where $H(\bm{\Gamma}^p_{\text{min}}) \in \mathbb{R}^{2p \times 2p}$, $v(l,k) \in \mathbb{R}^{2p \times 2}$ and, $h(l,k) \in \mathbb{R}^{2 \times 2} $. It is important to note that the determinant of the Hessian at the point $\bm{\Gamma}^{p+1}_{\text{TS}}(l, k)$ remains unchanged by such reordering of rows and columns. To see this, recall that switching two rows or columns causes the determinant to switch signs. Since we switch $x$ rows and $x$ columns, we realize that the overall sign does not change after all. In terms of matrix elements, $v(l,k) \in \mathbb{R}^{2p \times 2}$ reads
\begin{align*}
 v(l,k) =& 
\begin{pmatrix}
\partial_{\beta_1}\partial_{\beta_l} E(\bm{\beta}, \bm{\gamma})_{\big \vert \bm{\Gamma}^{p+1}_{\text{TS}}} & \partial_{\beta_1}\partial_{\gamma_k} E(\bm{\beta}, \bm{\gamma})_{\big \vert \bm{\Gamma}^{p+1}_{\text{TS}}} \\
\vdots & \vdots \\
\partial_{\beta_{l-1}}\partial_{\beta_l} E(\bm{\beta}, \bm{\gamma})_{\big \vert \bm{\Gamma}^{p+1}_{\text{TS}}} & \partial_{\beta_{l-1}}\partial_{\gamma_k} E(\bm{\beta}, \bm{\gamma})_{\big \vert \bm{\Gamma}^{p+1}_{\text{TS}}} \\
\partial_{\beta_{l+1}}\partial_{\beta_l} E(\bm{\beta}, \bm{\gamma})_{\big \vert \bm{\Gamma}^{p+1}_{\text{TS}}} & \partial_{\beta_{l+1}}\partial_{\gamma_k} E(\bm{\beta}, \bm{\gamma})_{\big \vert \bm{\Gamma}^{p+1}_{\text{TS}}} \\
\vdots & \vdots \\
\partial_{\beta_{p+1}}\partial_{\beta_l} E(\bm{\beta}, \bm{\gamma})_{\big \vert \bm{\Gamma}^{p+1}_{\text{TS}}} & \partial_{\beta_{p+1}}\partial_{\gamma_k} E(\bm{\beta}, \bm{\gamma})_{\big \vert \bm{\Gamma}^{p+1}_{\text{TS}}} \\
\partial_{\gamma_1}\partial_{\beta_l} E(\bm{\beta}, \bm{\gamma})_{\big \vert \bm{\Gamma}^{p+1}_{\text{TS}}} & \partial_{\gamma_1}\partial_{\gamma_k} E(\bm{\beta}, \bm{\gamma})_{\big \vert \bm{\Gamma}^{p+1}_{\text{TS}}} \\
\vdots & \vdots \\
\partial_{\gamma_{k-1}}\partial_{\beta_l} E(\bm{\beta}, \bm{\gamma})_{\big \vert \bm{\Gamma}^{p+1}_{\text{TS}}} & \partial_{\gamma_{k-1}}\partial_{\gamma_k} E(\bm{\beta}, \bm{\gamma})_{\big \vert \bm{\Gamma}^{p+1}_{\text{TS}}} \\
\partial_{\gamma_{k+1}}\partial_{\beta_l} E(\bm{\beta}, \bm{\gamma})_{\big \vert \bm{\Gamma}^{p+1}_{\text{TS}}} & \partial_{\gamma_{k+1}}\partial_{\gamma_k} E(\bm{\beta}, \bm{\gamma})_{\big \vert \bm{\Gamma}^{p+1}_{\text{TS}}} \\
\vdots & \vdots \\
\partial_{\gamma_{p+1}}\partial_{\beta_l} E(\bm{\beta}, \bm{\gamma})_{\big \vert \bm{\Gamma}^{p+1}_{\text{TS}}} & \partial_{\gamma_{p+1}}\partial_{\gamma_k} E(\bm{\beta}, \bm{\gamma})_{\big \vert \bm{\Gamma}^{p+1}_{\text{TS}}} 
\end{pmatrix},
\end{align*}
while $h(l,k) \in \mathbb{R}^{2 \times 2}$ becomes
\begin{align*}
h(l,k) &=& \begin{pmatrix}
\partial_{\beta_l}\partial_{\beta_l} E(\bm{\beta}, \bm{\gamma})_{\big \vert \bm{\Gamma}^{p+1}_{\text{TS}}} & \partial_{\beta_l}\partial_{\gamma_k} E(\bm{\beta}, \bm{\gamma})_{\big \vert \bm{\Gamma}^{p+1}_{\text{TS}}} \\
\partial_{\beta_l}\partial_{\gamma_k} E(\bm{\beta}, \bm{\gamma})_{\big \vert \bm{\Gamma}^{p+1}_{\text{TS}}} & \partial_{\gamma_k}\partial_{\gamma_k} E(\bm{\beta}, \bm{\gamma})_{\big \vert \bm{\Gamma}^{p+1}_{\text{TS}}}
\end{pmatrix}.
\end{align*}

Our goal is to restrict the properties of the Hessian~(\ref{eq:Hessian-big}) using the fact that the Hessian at circuit depth $p$ is a positive-definite matrix, a consequence of the fact that we start at a local minimum $\bm{\Gamma}^p_{\text{min}}$. To this end, we use a powerful theorem from matrix analysis.
\begin{theorem}[Eigenvalue interlacing theorem~\cite{matrix_book} (Theorem~4 on page 117)]
\label{th:eigInterlacing}
Let $A \in \mathbb{R}^{n \times n}$ be a symmetric matrix and $B\in \mathbb{R}^{m\times m}$ with $m<n$ be a principal submatrix (obtained by removing both the $i$-th column and $i$-th row for some values of $i$). Suppose $A$ has eigenvalues $\lambda_1 \leq \lambda_2 \leq \cdots \leq \lambda_n$ and $B$ has eigenvalues $\kappa_1 \leq \cdots \leq \kappa_m$. Then \begin{equation}
\lambda_k \leq \kappa_k \leq \lambda_{k+n-m},  
\end{equation}
for $ k={1,m}$.
\end{theorem}
The eigenvalue interlacing theorem relates the ordered set of Hessian eigenvalues $\{\lambda^{p+1}_{i}\}$ for QAOA$_{p+1}$ to the Hessian eigenvalues $\{\lambda^{p}_{i}\}$ of  QAOA$_p$ in the following way:
\begin{equation}
\lambda^{p+1}_k \leq \lambda^p_{k} \leq \lambda^{p+1}_{k+2}.
\end{equation}
Using the fact that $H_p(\bm{\Gamma}^p_{\text{min}})$ has $\lambda^p_{k}>0$ for all $k$, we see that the Hessian of QAOA$_{p+1}$  at point $\bm{\Gamma}^{p+1}_{\text{TS}}(l, k)$ has at most two negative eigenvalues, $\lambda^{p+1}_{1},\lambda^{p+1}_{2}<\lambda^p_1$, whereas $0<\lambda^p_1<\lambda^{p+1}_j$ for $j\geq 3$. 
In what follows we establish that among these two eigenvalues, exactly one is negative and the other one is positive. This is achieved by demonstrating that the full Hessian matrix has a negative determinant,
\begin{equation}\label{Eq:det-proof}
    \mathop{\rm det}H\big[\bm{\Gamma}^{p+1}_{\text{TS}}(l, k)\big] < 0,
\end{equation} 
which rules out the possibility that the remaining eigenvalues $\lambda^{p+1}_{1,2}$ have the same sign (which would cancel in the determinant). 

Below we first prove Relation~\eqref{Eq:det-proof} for the cases where the insertion of the zeros is made at the first {\bf{(i)}} or at the last {\bf{(ii)}} layer of the unitary circuit. We then conclude by considering the general case {\bf{(iii)}}, where zeros are inserted in the ``bulk" of the unitary circuit. Moreover, whenever is clear from context, we will drop the indices $(l,k)$ for better readability. Furthermore, for all the cases considered below, we introduce a specific short-hand notation for the following second-order derivative
\begin{equation}
    b = \partial_{\beta_l} \partial_{\gamma_k} E(\bm{\beta}, \bm{\gamma})_{\big \vert \bm{\Gamma}^{p+1}_{\text{TS}}}.
\end{equation}
This matrix element will play a special role in the calculation of $\text{det} H(\bm{\Gamma}^{p+1}_{\text{TS}}(l,k))$. It is important to note, that while the specific expression of $b$ differs for all the stationary points in the set given by Eq.~\eqref{eq:ts2}, it has a non-zero value, $b\neq 0$. Indeed, below we express $b$ as an expectation value of a non-vanishing operator over the QAOA variational state, that is non-zero in the absence of special symmetries. 

\subsubsection{Case {\bf{(i)}}: $l=k=p+1$}
\label{Sec:case1}
The first step is to compute the matrix elements of $v(p+1,p+1)$. From now on we drop the quantifying index and simply write $v$ and $h$ to reduce notational overhead.
The first column of $v$ corresponds to  $v_{\beta_{j},\beta_{p+1}}=\partial_{\beta_{j}}\partial_{\beta_{p+1}}E(\bm{\beta}, \bm{\gamma})$ evaluated at the TS $\bm{\Gamma}^{p+1}_{\text{TS}}$: 
\begin{equation}
\label{eq:case1_vbeta}
\begin{split}
&\partial_{\beta_{j}}\partial_{\beta_{p+1}}E(\bm{\beta}, \bm{\gamma})_{\big \vert \bm{\Gamma}^{p+1}_{\text{TS}}}  =
\\
&\langle+|U_{\leq j}^{\dagger}[U_{>j}^{\dagger}[H_{B},H_{C}]U_{>j},H_{B}]U_{\leq j}|+\rangle = a_j,
\end{split}
\end{equation}
where we introduced the short-hand notation $a_j$ for better readability.
Analogously, considering matrix elements of the form $v_{\gamma_{j},\beta_{p+1}}=\partial_{\gamma_{j}}\partial_{\beta_{p+1}}E(\bm{\beta}, \bm{\gamma})$, we obtain
\begin{multline}
\label{eq:case1_vgamma}
\partial_{\gamma_{j}}\partial_{\beta_{p+1}}E(\bm{\beta}, \bm{\gamma})_{\big \vert \bm{\Gamma}^{p+1}_{\text{TS}}}= \\
\langle +|U_{<j}^{\dagger}[U_{\geq j}^{\dagger}[H_{B},H_{C}]U_{\geq j},H_{C}]U_{<j}|+\rangle = a_{p+1+j}.
\end{multline}
Evaluating the second derivatives on Eq.~\eqref{eq:case1_vbeta} and Eq.~\eqref{eq:case1_vgamma} at $j=p+1$ corresponds to the first column of the $2 \times 2$ matrix $h$. In particular, evaluating Eq.~\eqref{eq:case1_vbeta} at $j=p+1$ leads to $U_{>j}=\mathbb{I}$ and $U_{\leq j}=U$ which in turn implies that 
\begin{multline}
\partial_{\beta_{p+1}}^{2}E(\bm{\beta}, \bm{\gamma})_{\big \vert \bm{\Gamma}^{p+1}_{\text{TS}}}=
\\
\langle \bm{\Gamma}^p_{\text{min}}|[[H_{B},H_{C}],H_{B}]|\bm{\Gamma}^p_{\text{min}}\rangle=a_{p+1}.
\end{multline}
Note that above we used $U_{>p+1}=\mathbb{I}$. This is because when the derivative is taken with respect to the last layer ($p+1$) of the unitary circuit, there is no unitary to the left of it which, in the notation introduced on Eq.\eqref{eq:grad_state} is equivalent to $U_{>p+1}=\mathbb{I}$. Doing the same on Eq.~\eqref{eq:case1_vgamma} gives
\begin{multline}
    \label{eq:case1_b}
    \partial_{\gamma_{p+1}}\partial_{\beta_{p+1}}E(\bm{\beta}, \bm{\gamma})_{\big \vert \bm{\Gamma}^{p+1}_{\text{TS}}}=\\
    \langle \bm{\Gamma}^p_{\text{min}}|[[H_{B},H_{C}],H_{C}]|\bm{\Gamma}^p_{\text{min}}\rangle=b.   
\end{multline}
Finally, let us look at the matrix elements of the form $v_{\beta_j,\gamma_{p+1}}=\partial_{\beta_j}\partial_{\gamma_{p+1}}E(\vec{\beta},\vec{\gamma})$ and analogously $v_{\gamma_j,\gamma_{p+1}}$, corresponding to the second column of $v$.
Let us first inspect $\partial_{\gamma_{p+1}}E(\vec{\beta},\vec{\gamma})$:
\begin{multline}
\partial_{\gamma_{p+1}}E(\bm{\beta}, \bm{\gamma})=\\
\mathrm{i} \langle+|U_{<p+1}^{\dagger}[H_{C},U_{p+1}^{\dagger}H_{C}U_{p+1}]U_{< p+1}|+\rangle.
\end{multline}
When evaluated at point $\bm{\Gamma}^{p+1}_{\text{TS}}$, we obtain that $[H_{C},U_{p+1}^{\dagger}H_{C}U_{p+1}]=0$ since $U_{p+1}=\mathbb{I}$ and $H_C$ commutes with itself. Hence, we see that as long as the second derivative is taken with respect to an element ($\beta$ or $\gamma$) at index $j<p+1$ the final result will be zero. As we already saw in Eq.~\eqref{eq:case1_b}, $\partial_{\gamma_{p+1}}\partial_{\beta_{p+1}}E(\bm{\beta},\bm{\gamma})$
is equal to $b$. Using similar arguments, we show
that $\partial_{\gamma_{p+1}}\partial_{\gamma_{p+1}}E(\bm{\beta},\bm{\gamma})=0$ which corresponds to the $h_{\gamma_{p+1}, \gamma_{p+1}}$ matrix element of $h$.
We are then ready to construct the Hessian at the TS under consideration:
\begin{equation}
H(\bm{\Gamma}_{\text{TS}}^{p+1})=\begin{pmatrix}H(\bm{\Gamma}^{p}_{\text{min}}) & v\\
v^{T} & h
\end{pmatrix},
\end{equation}
with 
\begin{equation}
v^{T}=\begin{pmatrix}
a_1 & \cdots & a_{2p+1}\\
0 & \cdots & 0
\end{pmatrix} \quad \text{and} \quad 
h = \begin{pmatrix}
a_{p+1} & b\\
b & 0
\end{pmatrix}.
\end{equation}
Using the expression for the determinant of a block matrix~\cite{matrix_book}
\begin{equation}
\label{eq:detBlock}
    \mathop{{\rm {det}}}\begin{pmatrix}A & B\\
C & D
\end{pmatrix}=\mathop{{\rm {det}}}(A)\mathop{{\rm {det}}}(D-CA^{-1}B),
\end{equation}
we  rewrite the determinant of the full Hessian as follows
\begin{multline}
\label{eq:hessianDet_case1}
{\rm det}\big[H(\Gamma_{TS}^{p+1})\big]=
\\
{\rm det}\left(\begin{array}{cc}
a_{p+1} & b\\
b & 0
\end{array}\right){\rm det}\big[H(\Gamma_{{\rm min}}^{p})-v h^{-1}v^T\big] \\
=-b^{2} {\rm det}\big[H(\Gamma_{{\rm min}}^{p})\big].
\end{multline}
We used that $v h^{-1}v^T=0$ in the last line. We then see that as long as $b \neq 0$ the determinant of the Hessian at the TS is negative, $\text{det}[H(\Gamma_{TS}^{p+1})]<0$. The explicit expression~(\ref{eq:case1_b}) for $b$ relates it to the expectation value of the commutator $[[H_B,H_C],H_C]$ over the variational wave function. Since this commutator is a non-vanishing operator, its expectation value is generically non-zero, $b\neq0$.  This concludes the proof of Theorem~\ref{th:ts} for the case when zeros are inserted at the last layer of the unitary circuit. 

\subsubsection{Case {\bf{(ii)}}: $l= k = 1$}
As before, we focus on computing the matrix elements of $v=v(1,1)$ and $h=h(1,1)$. Starting from the first column of $v$, with matrix elements $v_{\beta_j, \beta_1}$ and $v_{\gamma_j, \beta_1}$ for $j\in [2, p+1]$ we find
\begin{equation}
\label{eq:case2_vbeta}
    \begin{split}
    &\partial_{\beta_j}\partial_{\beta_1} E(\bm{\beta}, \bm{\gamma})_{\big \vert \bm{\Gamma}_{\text{TS}}^{p+1} } =\\
    &\langle +| [H_B, U^\dagger_{\leq j}[U^\dagger_{>j}H_C U_{>j}, H_B]U_{\leq j}]|+\rangle = 0, \\
    &\partial_{\gamma_j}\partial_{\beta_1} E(\bm{\beta}, \bm{\gamma})_{\big \vert \bm{\Gamma}_{\text{TS}}^{p+1} } =\\
    &\langle +| [H_B, U^\dagger_{<j}[U^\dagger_{\geq j}H_CU_{\geq j}, H_C]U_{<j}]|+\rangle = 0.
    \end{split}
\end{equation}

Moving onto the second column of $v$, with matrix elements $v_{\beta_j, \gamma_1}$ and $v_{\gamma_j, \gamma_1}$ for $j\in [2,p+1]$ we obtain
\begin{equation}
\label{eq:case2_vgamma}
    \begin{split}
    &\partial_{\gamma_1}\partial_{\beta_j} E(\bm{\beta}, \bm{\gamma})_{\big \vert \bm{\Gamma}_{\text{TS}}^{p+1} } = \\
    &\langle +|[H_C, U^\dagger_{\leq j}[U^\dagger_{>j}H_C U_{>j}, H_B]U_{\leq j}]|+\rangle = c_j, \\
    &\partial_{\gamma_j}\partial_{\gamma_1} E(\bm{\beta}, \bm{\gamma})_{\big \vert \bm{\Gamma}_{\text{TS}}^{p+1} } = \\
    &\langle +| [H_c, U^\dagger_{<j}[U^\dagger_{\geq j} H_C U_{\geq j}, H_C]U_{< j}]|+\rangle = c_{p+1+j}
    \end{split}
\end{equation}
where for better readability we introduced the short-hand notation $c_j$ with $j \in [2,p]$. Finally, evaluating the above expressions Eq.~\eqref{eq:case2_vbeta} and Eq.~\eqref{eq:case2_vgamma} at $j=1$ leads to the matrix elements of the $2\times 2$ matrix $h$. Altogether, we find
\begin{equation*}
v^{T}(1, 1)=\begin{pmatrix}
0 & \cdots & 0 \\
c_1 & \cdots & c_{2p+2}
\end{pmatrix},
\quad
h(1, 1)=\begin{pmatrix}0 & b\\
b & c_{p+2}
\end{pmatrix},
\end{equation*}
where 
\begin{equation}
    \label{eq:case2_b}
    b = \langle +|[H_C, [U^\dagger H_C U, H_B]]|+ \rangle
\end{equation}
and the value of $c_{p+2}$ follows from evaluating Eq.~\eqref{eq:case2_vgamma} at $j=1$. 

Invoking once again the expression for the determinant of a block matrix Eq.~\eqref{eq:detBlock} we get
\begin{multline}
   \text{det}\big[H(\Gamma_{TS}^{p+1})\big]=\text{det}\big[H(\Gamma_{{\rm min}}^{p})) \text{det}(h + v^T  H(\Gamma_{{\rm min}}^{p})  v\big) \\
   =\text{det}\bigg[\left(\begin{array}{cc}
   0 & b \\
   b & c_{p+2}
   \end{array}\right) + 
   \left(\begin{array}{cc}
   0 & 0 \\
   0 & \text{const}
   \end{array}\right) \bigg]\text{det}\big[H(\Gamma_{{\rm min}}^{p})\big], \\
=-b^{2}\text{det}\big[H(\Gamma_{{\rm min}}^{p})\big]. 
\end{multline}
Using that the point $\bm{\Gamma}_{{\rm min}}^{p}$ is a local minimum (with the Hessian being non-singular), we see that as long as $b\neq 0$ the determinant of the Hessian at the TS is negative. The fact that the parameter $b$ in Eq.~(\ref{eq:case2_b}) is non-vanishing can be inferred from the similar argument to the one used at the end of Appendix~\ref{Sec:case1}
\subsubsection{Case {\bf{(iii)}}: $l,k \in {2, p}$}
So far we have proven that when the zeros insertion is made at the initial (I) or last (II) layer of the unitary circuit the corresponding points $\bm{\Gamma}^{p+1}_{\text{TS}}$ of QAOA$_{p+1}$  are TS. In both cases, we proved that the determinant of the Hessian of QAOA$_{p+1}$ at the given points is negative. In order to do this, we used that one of the columns of the $2p \times 2$ matrix $v$ was zero which greatly simplified the computation of the determinant. In what follows, we show that these simplifications, unfortunately, do not occur when the zeros insertion is made in the bulk of the unitary circuits. However, we instead observe that the matrix $v(l,k)$ is constructed by taking the $l$-th ($\beta_l$) and $p+1+k$-th ($\gamma_k$) columns of the Hessian of QAOA$_{p}$ at the local minimum $\bm{\Gamma}^p_{\text{min}}$. This fact, together with the invariance of the determinant under linear operations performed on rows or columns leads to the desired result.

We begin by explicitly computing the matrix elements of $h(l,k)$ and $v(l,k)$ and then relating them to matrix elements of the Hessian $H(\bm{\Gamma^p_{\text{min}}})$. For the sake of concreteness, we focus on the particular case of symmetric TS, i.e.\ $k=l$. 
The other case, i.e. $k=l+1$ can be covered by an analogous chain of arguments. 
As before, in what follows we drop the quantifying indices for better readability. Starting from $h$, we obtain
\begin{multline}
    h = \begin{pmatrix}
\partial_{\beta_l}\partial_{\beta_l} E(\bm{\beta}, \bm{\gamma})_{\big \vert \bm{\Gamma}^{p+1}_{\text{TS}}} & \partial_{\beta_l}\partial_{\gamma_l} E(\bm{\beta}, \bm{\gamma})_{\big \vert \bm{\Gamma}^{p+1}_{\text{TS}}} \\
\partial_{\beta_l}\partial_{\gamma_l} E(\bm{\beta}, \bm{\gamma})_{\big \vert \bm{\Gamma}^{p+1}_{\text{TS}}} & \partial_{\gamma_l}\partial_{\gamma_l} E(\bm{\beta}, \bm{\gamma})_{\big \vert \bm{\Gamma}^{p+1}_{\text{TS}}}
\end{pmatrix}  \\
=
\begin{pmatrix}
\partial^2_{\beta_{l-1}}E(\bm{\beta}, \bm{\gamma})_{\big \vert \bm{\Gamma}^p_{\text{min}}} & 
b \\
b & \partial^2_{\gamma_{l}}E(\bm{\beta}, \bm{\gamma})_{\big \vert \bm{\Gamma}^p_{\text{min}}}
\end{pmatrix} \\
= \begin{pmatrix}
H(\bm{\Gamma}^p_{\text{min}})_{\beta_{l-1},\beta_{l-1}} & b \\
b & H(\bm{\Gamma}^p_{\text{min}})_{\gamma_{l},\gamma_{l}}
\end{pmatrix},
\end{multline}
where
\begin{equation}
\label{eq:case3_b}
    b = \langle +|U_{\leq l-1}^{\dagger}[H_{C},[H_{B},U_{>l-1}^{\dagger}H_{C}U_{>l-1}]]U_{\le l-1}|+\rangle.
\end{equation}

One might be tempted by looking at the properties listed in Eq.~\eqref{eq:grad_TS} to relate $\partial_{\beta_l}\partial_{\gamma_l} E(\bm{\beta}, \bm{\gamma})_{\big \vert \bm{\Gamma}^{p+1}_{\text{TS}}}$ to $\partial_{\beta_{l-1}}\partial_{\gamma_l} E(\bm{\beta}, \bm{\gamma})_{\big \vert \bm{\Gamma}^{p}_{\text{min}}}$. However, upon closer inspection, we can see that these are not the same. More specifically, we get
\begin{multline}
    \partial_{\beta_{l-1}}\partial_{\gamma_l} E(\bm{\beta}, \bm{\gamma})_{\big \vert \bm{\Gamma}^{p}_{\text{min}}} = 
    \\
    \langle +| U^\dagger_{\leq l-1} [H_B, [H_C, U^\dagger_{> l-1} H_C U_{> l-1}]] U_{\leq l-1} |+ \rangle.
\end{multline}
Comparing the above expression with Eq.~\eqref{eq:case3_b} we realize that although not equal, they are related via the Jacobi identity 
\begin{equation}
    [A, [B, C]] + [B, [C, A]] + [C, [A, B]] = 0, 
\end{equation}
for operators $A,B$ and $C$. More specifically, we obtain
\begin{equation}
\begin{split}
    &b - \partial_{\beta_{l-1}}\partial_{\gamma_l} E(\bm{\beta}, \bm{\gamma})_{\big \vert \bm{\Gamma}^p_{\text{min}}} = \\
    &\langle + |U_{\leq l-1}^{\dagger}[U_{>l-1}^{\dagger}H_{C}U_{>l-1}, [H_{B}, H_{C}]]U_{\le l-1}|+\rangle = \bar{b}.
\end{split}
\end{equation}

Considering now the matrix elements of $v$ we get
\begin{multline}
    v = 
\begin{pmatrix}
\partial_{\beta_1}\partial_{\beta_{l-1}} E(\bm{\beta}, \bm{\gamma})_{\big \vert \bm{\Gamma}^{p}_{\text{min}}} & \partial_{\beta_1}\partial_{\gamma_l} E(\bm{\beta}, \bm{\gamma})_{\big \vert \bm{\Gamma}^{p}_{\text{min}}} \\
\vdots & \vdots \\
\partial_{\beta_{l-1}}\partial_{\beta_{l-1}} E(\bm{\beta}, \bm{\gamma})_{\big \vert \bm{\Gamma}^{p}_{\text{min}}} & \partial_{\beta_{l-1}}\partial_{\gamma_l} E(\bm{\beta}, \bm{\gamma})_{\big \vert \bm{\Gamma}^{p}_{\text{min}}} \\
\partial_{\beta_{l}}\partial_{\beta_{l-1}} E(\bm{\beta}, \bm{\gamma})_{\big \vert \bm{\Gamma}^{p}_{\text{min}}} & \partial_{\beta_{l}}\partial_{\gamma_l} E(\bm{\beta}, \bm{\gamma})_{\big \vert \bm{\Gamma}^{p}_{\text{min}}} \\
\vdots & \vdots \\
\partial_{\beta_{p}}\partial_{\beta_{l-1}} E(\bm{\beta}, \bm{\gamma})_{\big \vert \bm{\Gamma}^{p}_{\text{min}}} & \partial_{\beta_{p}}\partial_{\gamma_l} E(\bm{\beta}, \bm{\gamma})_{\big \vert \bm{\Gamma}^{p}_{\text{min}}} \\
\partial_{\gamma_1}\partial_{\beta_{l-1}} E(\bm{\beta}, \bm{\gamma})_{\big \vert \bm{\Gamma}^{p}_{\text{min}}} & \partial_{\gamma_1}\partial_{\gamma_l} E(\bm{\beta}, \bm{\gamma})_{\big \vert \bm{\Gamma}^{p}_{\text{min}}} \\
\vdots & \vdots \\
\partial_{\gamma_{l-1}}\partial_{\beta_{l-1}} E(\bm{\beta}, \bm{\gamma})_{\big \vert \bm{\Gamma}^{p}_{\text{min}}} & \partial_{\gamma_{l-1}}\partial_{\gamma_l} E(\bm{\beta}, \bm{\gamma})_{\big \vert \bm{\Gamma}^{p}_{\text{min}}} \\
\partial_{\gamma_{l}}\partial_{\beta_{l-1}} E(\bm{\beta}, \bm{\gamma})_{\big \vert \bm{\Gamma}^{p}_{\text{min}}} & \partial_{\gamma_{l}}\partial_{\gamma_l} E(\bm{\beta}, \bm{\gamma})_{\big \vert \bm{\Gamma}^{p}_{\text{min}}} \\
\vdots & \vdots \\
\partial_{\gamma_{p}}\partial_{\beta_{l-1}} E(\bm{\beta}, \bm{\gamma})_{\big \vert \bm{\Gamma}^{p}_{\text{min}}} & \partial_{\gamma_{p}}\partial_{\gamma_l} E(\bm{\beta}, \bm{\gamma})_{\big \vert \bm{\Gamma}^{p}_{\text{min}}} 
\end{pmatrix} \\
=
\begin{pmatrix}
H(\bm{\Gamma^p_{\text{min}}})_{\beta_1, \beta_{l-1}} & H(\bm{\Gamma^p_{\text{min}}})_{\beta_1, \gamma_{l}} \\
\vdots & \vdots \\
H(\bm{\Gamma^p_{\text{min}}})_{\beta_p, \beta_{l-1}} &
H(\bm{\Gamma^p_{\text{min}}})_{\beta_p, \gamma_{l}} \\
H(\bm{\Gamma^p_{\text{min}}})_{\gamma_1, \beta_{l-1}} &
H(\bm{\Gamma^p_{\text{min}}})_{\gamma_1, \gamma_{l}} \\
\vdots & \vdots \\
H(\bm{\Gamma^p_{\text{min}}})_{\gamma_p, \beta_{l-1}} &
H(\bm{\Gamma^p_{\text{min}}})_{\gamma_p, \gamma_{l}}
\end{pmatrix}.
\end{multline}

Hence, we find that the $2p \times 2$ rectangular matrix $v$ corresponds to the matrix formed by taking columns $H(\bm{\Gamma^p_{\text{min}}})_{m, \beta_{l-1}}$ and $H(\bm{\Gamma^p_{\text{min}}})_{m, \gamma_{l}}$ with $m=1,\ldots, 2p$ of $H(\bm{\Gamma^p_{\text{min}}})$. Using this result and the fact that the determinant is invariant under linear operations performed on rows or columns, we get that
\begin{equation}
    \text{det}(H(\bm{\Gamma^{p+1}_{\text{TS}}})) = \text{det}\begin{pmatrix}
   H(\bm{\Gamma}^p_{\text{min}}) & v(l,k) \\
  0 & \overline{h}(l,l)
  \end{pmatrix},
\end{equation}
where we subtracted rows $H(\bm{\Gamma}^p_{\text{min}})_{\beta_{l-1}, m }$ and  
$H(\bm{\Gamma}^p_{\text{min}})_{\gamma_{l}, m}$ with $m=1, \ldots, 2p$ from $v^T$, and introduced
\begin{equation}
    \overline{h} = \begin{pmatrix}
    0 & \bar b \\
    \bar b & 0
    \end{pmatrix},
\end{equation}
Using once again the expression for the determinant of a block matrix Eq.~\eqref{eq:detBlock}, and the fact that $\text{det}( \overline{h}(l,l)) = -\bar{b}^2$ is negative ($\bar{b}\neq 0$ due to similar argument as in Appendix~\ref{Sec:case1}) we obtain
\begin{equation}
    \text{det}\big[H(\bm{\Gamma}^{p+1}_{\text{TS}})\big] = -\bar{b}^2 \text{det}\big[H(\bm{\Gamma}^{p}_{\text{min}})\big] < 0,
\end{equation}
concluding our proof for the general TS. 

\section{Counting of unique minima}

The number of minima found in the initialization graph construction presented in the main text, na\"ively scales as $N_\text{min}(p)=2^{p-1}p!$. This follows from our recursive construction. Each local minimum of  QAOA$_p$ is used to construct $p+1$ symmetric TS and for each TS we then find two new minima of QAOA$_{p+1}$, see Figs.~\ref{fig:1} and~\ref{fig:2}. This factorial growth is, however, only sustained if every TS produces two new minima that are all distinct from each other.
Numerically, we find that this is not the case and that the number of unique minima is significantly smaller. The increase in the number of unique minima is consistent with an exponential dependence proportional to  $e^{\kappa p}$ [we find that $N_\text{min}(p)$ can be approximated as $N_\text{min}(p)\approx 0.19 e^{0.98 p}$]. However, the limited range of $p$ does not allow us to  completely rule out factorial growth, see Fig.~\ref{fig:supp0a}. The much smaller number of unique minima, compared to the na\"ive counting demonstrates that different TS often lead to similar minima, as illustrated in Fig.~\ref{fig:4}. 

\begin{figure}[b]
    \centering
    \includegraphics[width=\columnwidth]{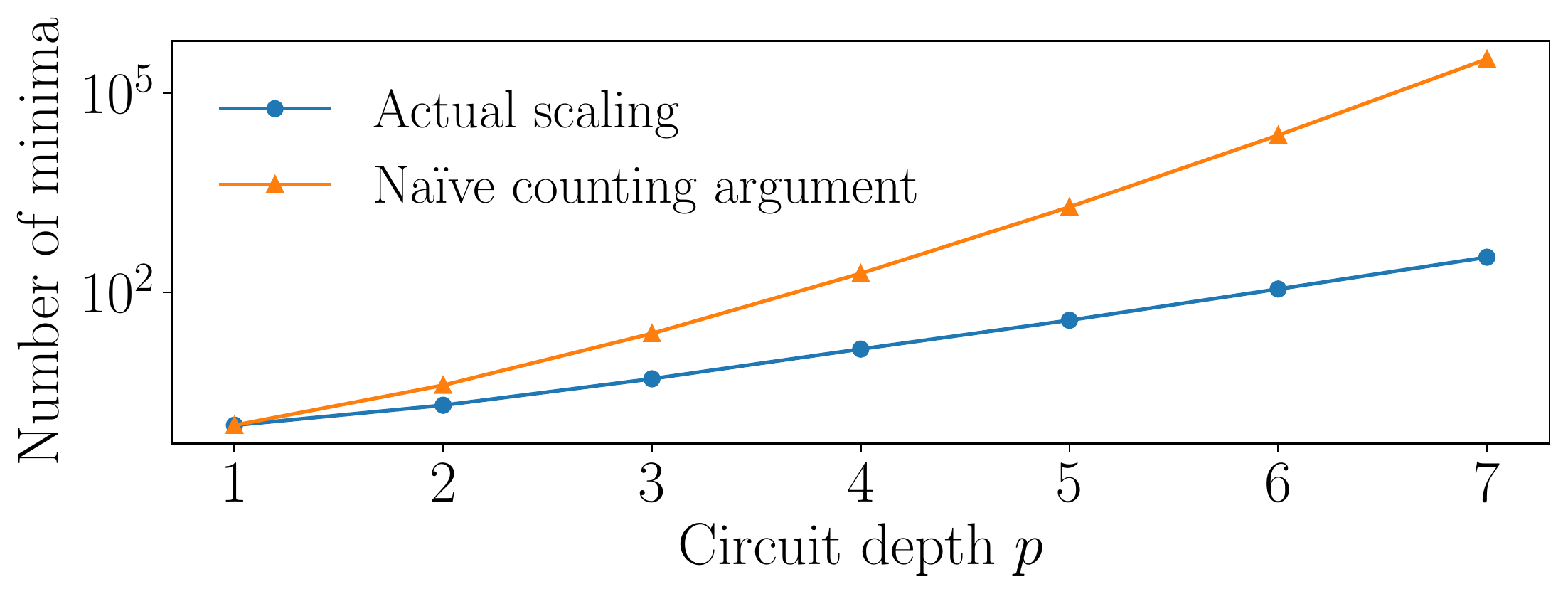}
    \caption{Number of minima found in the initialization graph in Fig.~\ref{fig:2} with system size $n=10$. The orange line describes a na\"ive counting argument ($2^{p-1}p!$) while the blue line lists the actual number of distinct minima that can be approximated as $0.19\, e^{0.98 p}$.
    }
    \label{fig:supp0a}
\end{figure}

\section{Properties of the index-1 direction}\label{app:index_1}

The index-1 direction is the direction of negative curvature at a TS in a QAOA$_{p+1}$ which we use to find two new minima in QAOA$_{p+1}$, as illustrated in Fig.~\ref{fig:2}(a). The index-1 direction is obtained by finding the eigenvector corresponding to the unique negative eigenvalue of the Hessian, $H(\bm{\Gamma}^{p+1}_{TS})$. Numerically we showed in Fig.~\ref{fig:2}(b) that optimization initialized along the $\pm$ index-1 direction either heals or enhances the perturbation introduced by a creation of the TS from the local minima of QAOA$_p$. 

\begin{figure}[h]
    \centering
    \includegraphics[width=0.92\columnwidth]{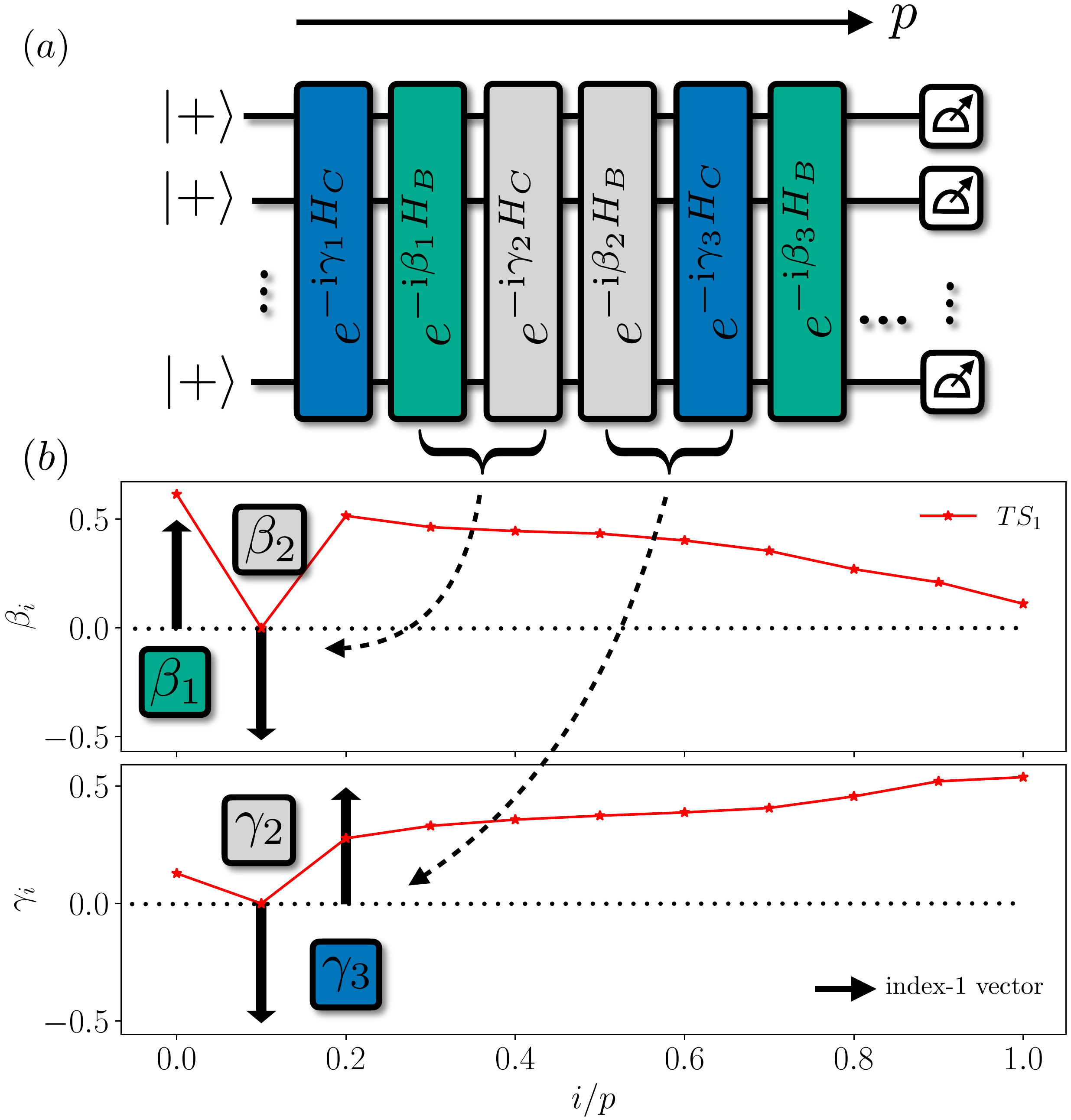}
    \caption{(a) Illustration of the circuit implementing the QAOA at a TS. Gray gates correspond to the zero insertion. The index-1 direction has mainly weight at the position of the zeros as well as the two adjacent gates. (b) Numerical example of the index-1 vector and the QAOA parameter pattern at the TS. Arrows correspond to the magnitude and sign of the entries in the index-1 direction. Only entries at $\beta_1, \beta_2, \gamma_2 \, \text{and} \, \gamma_3$ have a large magnitude, all other entries are nearly zero.
    }
    \label{fig:supp0}
\end{figure}

Interestingly, we find that the index-1 vector has dominant components at positions where zero angles were inserted as well as the positions of adjacent angles. In contrast, all other components of the index-1 vector have nearly zero weight, as illustrated in Fig.~\ref{fig:supp0}. The large contribution along the component corresponding to the zero insertion can be physically motivated by the fact that the gate with the zero parameter does initially not have any effect for driving the initial state $\ket{+}^{\otimes n}$ towards the ground state of $H_C$. Hence, the energy can be lowered by `switching on' the action of this gate by moving the value of the corresponding  variational angle away from zero. Interestingly, we see that the neighboring gates with non-zero parameters are also changed along the index-1 direction. The next nearest neighboring gates appear to be not involved in this process. We note that this numerical observation allows to \textit{a priori} guess the index-1 direction without having to diagonalize the Hessian $H(\bm{\Gamma}^{p+1}_{TS})$. This may be useful for the practical implementation of our initialization on available quantum computers.  

\section{Description of GREEDY algorithm}\label{app:greedy}

In the following, we provide a detailed description for the \textsc{GREEDY} QAOA initialization, as well as the \texttt{subroutines} required to implement the algorithm. To this end, we first provide a pseudo-code for a gradient-based QAOA parameter optimization routine. The algorithm is a so-called variational hybrid algorithm, which implies that the quantum computer is used in a closed feedback loop with a classical computer. There the quantum computer is used to implement a variational state and measure observables while the classical computer is used to keep track of the variational parameters and update them in order to minimize the energy expectation value. 
\begin{algorithm}[H]
\caption{\texttt{QAOA subroutine}}
\label{algo:wbp}
\begin{algorithmic}[1]
\State Given the circuit depth $p$, choose initial parameters $\bm{\Gamma}_{\text{init.}}^p=(\bm{\beta}_{\text{init.}}, \bm{\gamma}_{\text{init.}})$
\Repeat 
\State Implement $\ket{\bm{\beta}, \bm{\gamma}}$ on a quantum device
\State Estimate $E(\bm{\beta}, \bm{\gamma})=\langle \bm{\beta}, \bm{\gamma} |H_C| \bm{\beta}, \bm{\gamma} 
\rangle$ 
\State Estimate gradient $\nabla E(\bm{\beta}, \bm{\gamma})$
\State Update $(\bm{\beta}, \bm{\gamma})$ using gradient information
\Until{$E(\bm{\beta}, \bm{\gamma})$ has converged}
\State Return minimum $\bm{\Gamma}_{\min}^p$
\end{algorithmic}
\end{algorithm}

For very shallow circuit depths, such as $p=1$, the optimization landscape is sufficiently low dimensional and simple such that global optimization routines can be used to find the optimal parameters. One of the most straightforward global optimization routines is the so-called grid search. There, the parameters are initialized on a dense grid and a parameter optimization routine, such as the \texttt{QAOA sub-routine} is carried out for each point in the grid. Then, only the lowest energy local minimum is kept. 
\begin{algorithm}[H]
\caption{\texttt{grid search subroutine}}
\label{algo:wbp}
\begin{algorithmic}[1]
\State Given a circuit depth $p$, construct an evenly spaced grid on the fundamental region:
\begin{equation}\label{eq:symm}
 \beta_{i} \in \bigg[-\frac{\pi}{4},\frac{\pi}{4}\bigg]; 
 \ \,
 \gamma_1 \in \bigg(0,\frac{\pi}{4}\bigg),
 \ \, 
\gamma_{j} \in \bigg[-\frac{\pi}{4},\frac{\pi}{4}\bigg],
\end{equation}
with $i\in [1,p]$ and $j \in [2,p]$
\State \texttt{QAOA subroutine} initialized from each point in grid
\State Return local minimum with the lowest energy $\bm{\Gamma}_{\min}^{p}$
\end{algorithmic}
\end{algorithm}

Using the two \texttt{subroutines} presented above we can provide a detailed pseudo-code for the \textsc{Greedy} QAOA algorithm, see Fig.~\ref{fig:supp1a} for a visualization.
\begin{algorithm}[H]
\caption{\textsc{Greedy QAOA}}
\label{algo:greedy}
\begin{algorithmic}[1]
\State Choose maximum circuit depth $p_{\max}$
\State Choose small offset $\epsilon \ll 1$
\State Grid search for $p=1$ to find $\bm{\Gamma}_{\min}^{p=1}$ \Comment{See \texttt{grid search subroutine}} 
\Repeat 
\State Construct $p+1$ symmetric TS $\bm{\Gamma}_{TS}^{i, p+1}$ from $\bm{\Gamma}_{\min}^p$
\State Compute or approximate the index-1 unit vector $\bm{\hat{v}}$ for each TS
\State Construct points $\bm{\Gamma}_{\pm}^{i, p+1}=\bm{\Gamma}_{TS}^{i, p+1}\pm \epsilon \bm{\hat{v}}_i$ for each TS
\State Run QAOA init. from $\bm{\Gamma}_{\pm}^{i, p+1}$ \Comment{See \texttt{QAOA subroutine}}
\State Keep local minimum with the lowest energy $\bm{\Gamma}_{\min}^{p+1}$
\State $ p \longleftarrow p+1 $
\Until{$p=p_{\max}$}
\State Return minimum $\bm{\Gamma}_{\min}^{p=p_{\max}}$
\end{algorithmic}
\end{algorithm}
The index-1 direction $\bm{\hat{v}}_i$ can either be found explicitly by diagonalizing the Hessian matrix or using the heuristic approximation outlined in the previous section. While explicit diagonalization incurs classical computation costs that scale polynomially with $p$, and thus can be done efficiently, approximation to index-1 direction is expected to give similar performance of QAOA subroutine at a lower classical computational cost.  

\begin{figure}[t]
    \centering
    \includegraphics[width=0.94\columnwidth]{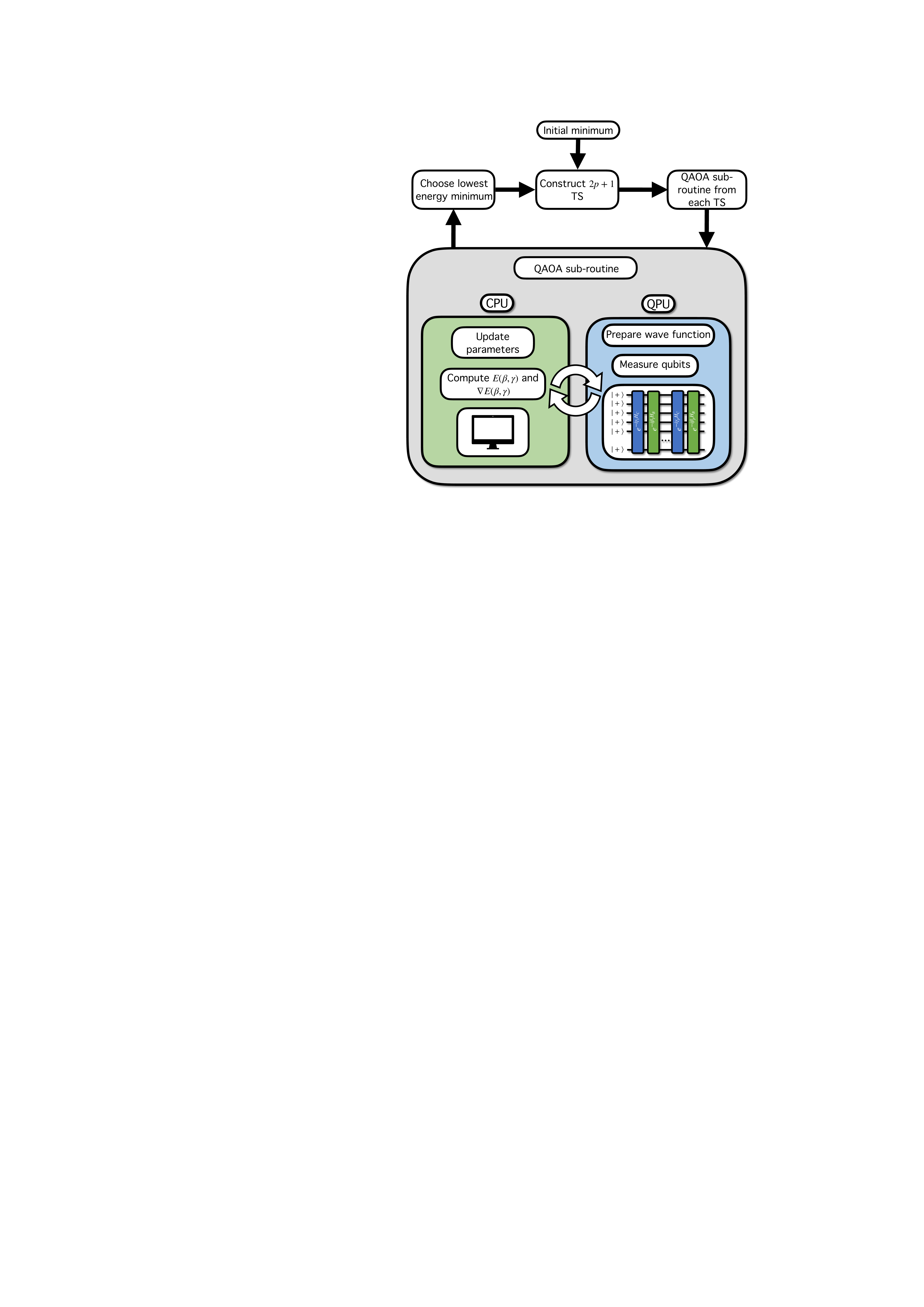}
    \caption{Flow diagram to visualize the \textsc{Greedy} QAOA initialization algorithm presented in Algorithm~\ref{algo:greedy}.}
    \label{fig:supp1a}
\end{figure}

\section{Additional graph ensembles and system size scaling \label{App:last}}

\begin{figure}[t]
    \centering
    \includegraphics[width=0.94\columnwidth]{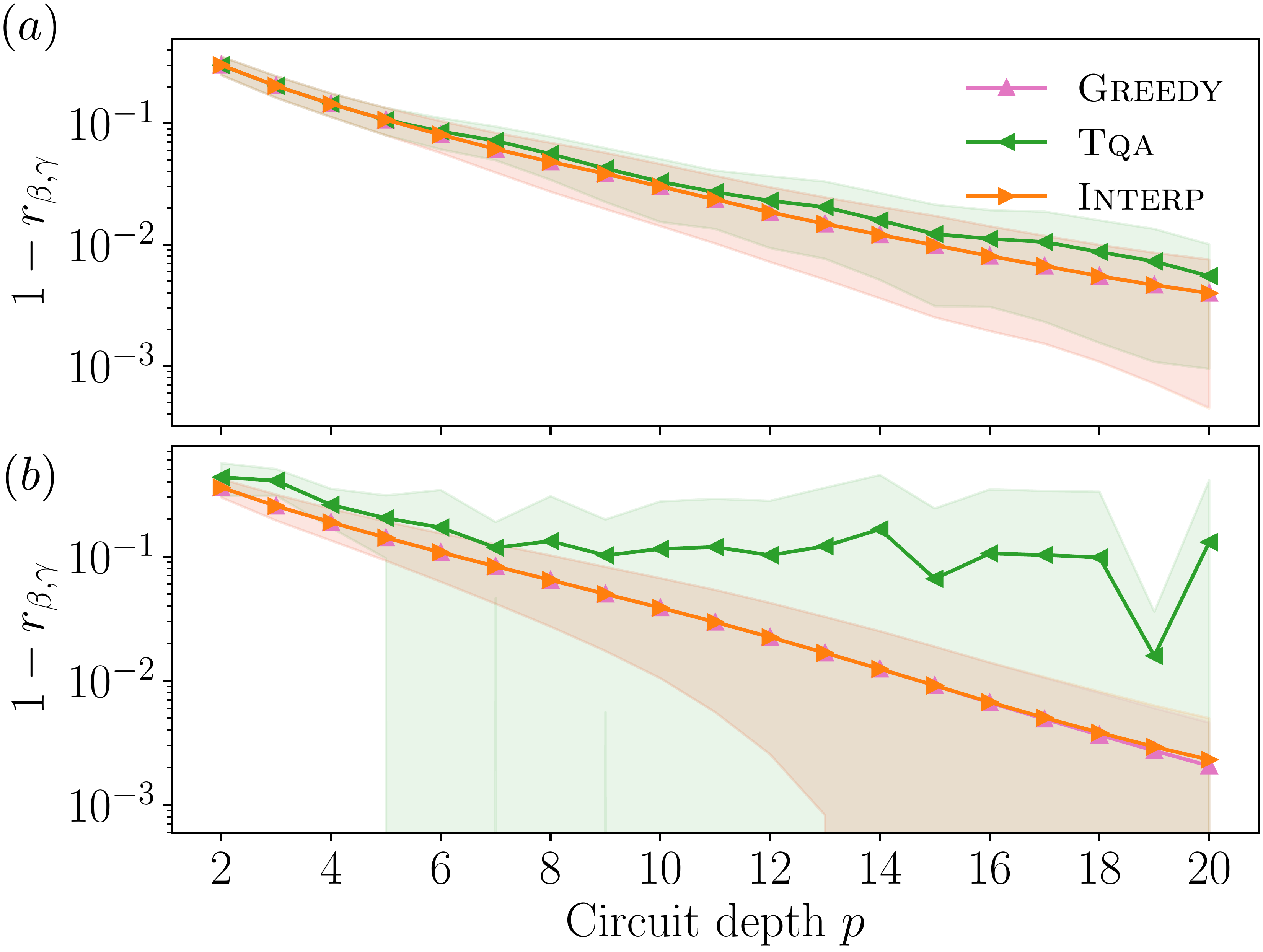}
    \caption{Performance comparison on (a) RWRG3 and (b) RERG with system size $n=10$. Data is averaged over $19$ non-isomorphic graphs.}
    \label{fig:supp1}
\end{figure}

In the main text, we numerically investigated the performance of our method on random 3-regular graphs (RRG3) with system size $n=10$. In the following, we present results for larger system sizes as well as two more graph types. Namely, weighted 3-random regular graphs (RWRG3) where the Hamiltonian is given by $H_C = \sum_{\langle i, j \rangle \in E} w_{ij} \sigma_i^z \sigma_j^z$ and $w_{ij}$ are random weights $w_{i j} \in [0, 1)$, as well as random Erd\H{o}s-R\'enyi graphs (RERG) with edge probability $p_E = 0.5$.

Fig.~\ref{fig:supp1} shows the performance comparison between \textsc{Greedy}, \textsc{TQA}, and \textsc{Interp} on RWRG3 and RERG. We can see that for RWRG3 the performance of the three methods is comparable, while for RERG the \textsc{TQA} performs worse that the other two methods. \textsc{Greedy} and \textsc{Interp} yield (nearly) the same performance for both graph ensembles on the system size that we considered ($n=10$).

\begin{figure}[b]
    \centering
    \includegraphics[width=0.94\columnwidth]{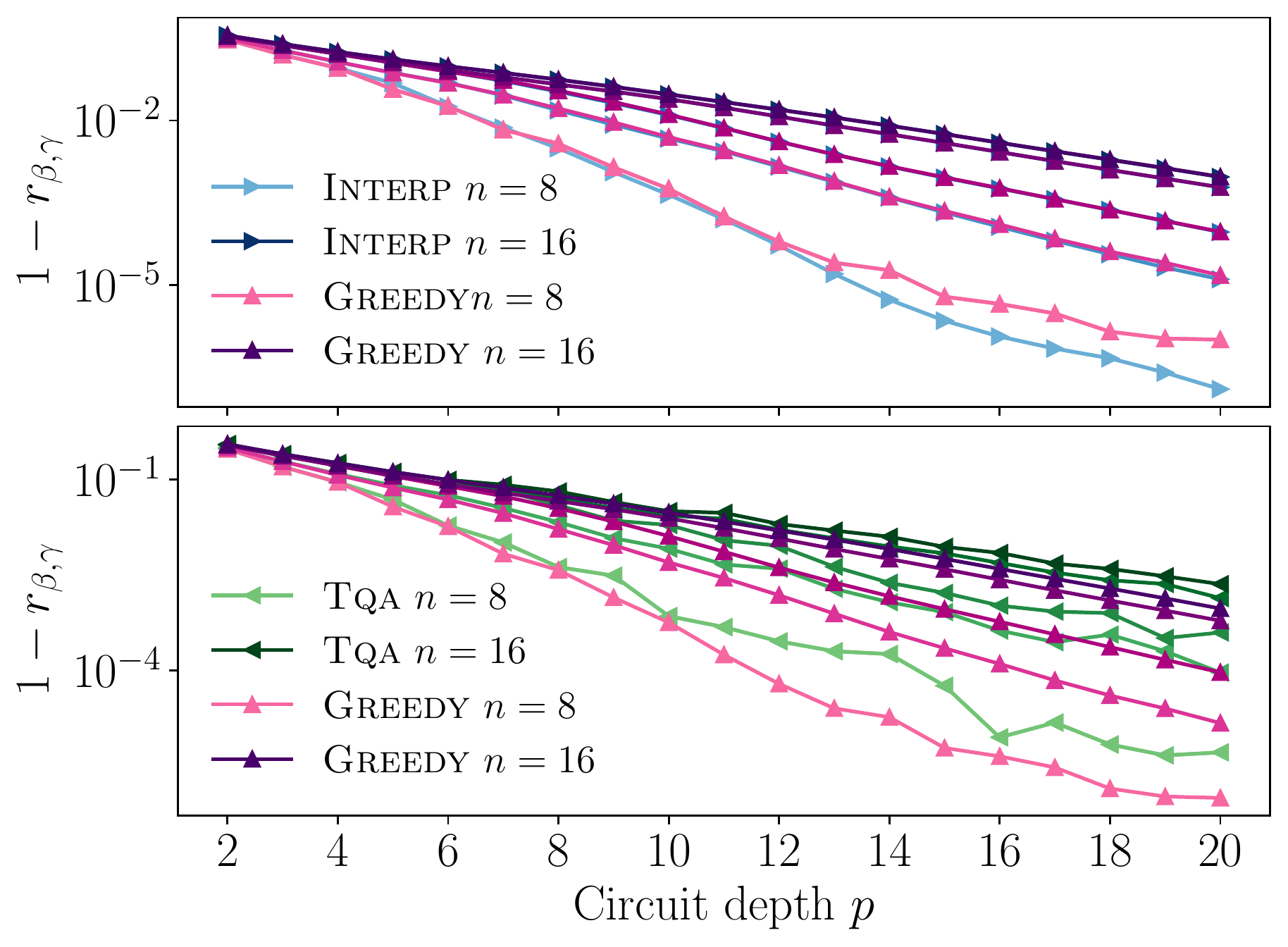}
    \caption{System size scaling for performance comparison on RRG3. Color shade indicates system size, light color is $n=8$ and dark color is $n=16$. System size changes in steps of two between those values. Data is averaged over $19$ non-isomorphic RRG3 graphs. }
    \label{fig:supp2}
\end{figure}

Fig.~\ref{fig:supp2} compares the performance for RRG3 with different system sizes. \textsc{Interp} and \textsc{Greedy} yield very similar performance for smaller system sizes ($n=8$ indicated by light color) while it yields the same performance for larger system sizes ($n=16$ indicated by dark color). \textsc{TQA} performs slightly worse than \textsc{Greedy} and \textsc{Interp} for all system sizes considered. We can furthermore see that gain in performance from every additional layer is becoming less for bigger system sizes. This is due to the fact that in order for the QAOA to ``see" the whole graph, a circuit depth $p$ scaling as $p \sim \log n$ is required~\cite{farhi2020needs}.

\end{document}